\documentclass[lettersize,journal]{IEEEtran}

\usepackage{amsmath,amssymb,amsfonts}
\usepackage{array}
\usepackage[caption=false,font=normalsize,labelfont=sf,textfont=sf]{subfig}
\usepackage{textcomp}
\usepackage{stfloats}
\usepackage{verbatim}
\usepackage{graphicx}
\usepackage{cite}
\usepackage{tabularx}
\usepackage{booktabs}

\def\BibTeX{{\rm B\kern-.05em{\sc i\kern-.025em b}\kern-.08em
    T\kern-.1667em\lower.7ex\hbox{E}\kern-.125emX}}

\usepackage{acronym}
\acrodef{5G}{the fifth generation}
\acrodef{MIMO}{multiple-input multiple-output}
\acrodef{SISO}{single-input single-output}
\acrodef{RF}{radio frequency}
\acrodef{LoS}{line-of-sight}
\acrodef{NLoS}{non-line-of-sight}
\acrodef{AoA}{angle-of-arrival}
\acrodef{AoD}{angle-of-departure}
\acrodef{UPA}{uniform planar array}
\acrodef{ARV}{array response vector}
\acrodef{EM}{electromagnetic}
\acrodef{MA}{movable antenna}
\acrodef{BS}{base station}
\acrodef{UE}{user equipment}
\acrodef{ERA}{electromagnetically reconfigurable antenna}
\acrodef{AWGN}{additive white Gaussian noise}
\acrodef{HFSS}{High-Frequency Structure Simulator}
\acrodef{WMMSE}{weighted minimum mean-square error}
\acrodef{BCD}{block coordinate descent}
\acrodef{ZF}{zero-forcing}
\acrodef{SNR}{Signal-to-Noise Ratio}
\acrodef{6G}{sixth generation}
\acrodef{3D}{three-dimensional}
\acrodef{2D}{two-dimensional}
\acrodef{RIS}{reconfigurable intelligent surface}
\acrodef{CSI}{channel state information}
\newtheorem{lemma}{\textbf{Lemma}}
\newtheorem{proposition}{\textbf{Proposition}}

\newtheorem{remark}{\textbf{Remark}}

\newenvironment{proof}{\textit{\textbf{Proof:}}}{\hfill$\square$}
\usepackage{color}

\newcommand{\violet}[1]{{\color{violet}{#1}}} 

\usepackage{tikz} 
\usepackage[utf8]{inputenc}
\usepackage{pgfplots}
\usepackage{tabu,longtable}
\usepackage{diagbox}
\usepackage{threeparttable}
\usepackage{makecell}
\usepackage{multirow} 
\usepackage{units} 		
\usepackage{bm} 		
\usepackage{amssymb} 	
\newcommand{\TT}{\mathsf{T}}
\newcommand{\HH}{\mathsf{H}}

\newcommand{\av}{{\bf a}}
\newcommand{\bv}{{\bf b}}
\newcommand{\cv}{{\bf c}}
\newcommand{\dv}{{\bf d}}

\newcommand{\fv}{{\bf f}}
\newcommand{\gv}{{\bf g}}

\newcommand{\nv}{{\bf n}}

\newcommand{\pv}{{\bf p}}

\newcommand{\rv}{{\bf r}}
\newcommand{\sv}{{\bf s}}

\newcommand{\uv}{{\bf u}}

\newcommand{\vv}{{\bf v}}
\newcommand{\xv}{{\bf x}}
\newcommand{\yv}{{\bf y}}

\newcommand{\Am}{{\bf A}}
\newcommand{\Bm}{{\bf B}}
\newcommand{\Cm}{{\bf C}}
\newcommand{\Dm}{{\bf D}}
\newcommand{\Em}{{\bf E}}
\newcommand{\Fm}{{\bf F}}
\newcommand{\Gm}{{\bf G}}
\newcommand{\Hm}{{\bf H}}

\newcommand{\Nm}{{\bf N}}

\newcommand{\Qm}{{\bf Q}}

\newcommand{\Um}{{\bf U}}
\newcommand{\Wm}{{\bf W}}

\newcommand{\Dt}{{\mathsf D}}

\newcommand{\gammav}{\hbox{\boldmath$\gamma$}}

\newcommand{\Gammam}{\hbox{\boldmath$\Gamma$}}

\usepackage{hyperref} 	
\usepackage{amsmath}

\usepackage{algorithm}
\usepackage{algpseudocode}

\makeatletter
\algnewcommand{\LineComment}[1]{\Statex \hskip\ALG@thistlm \(\triangleright\) #1}
\makeatother

\begin{document}
\bstctlcite{IEEEexample:BSTcontrol} 

\title{Tri-Hybrid Multi-User Precoding Using Pattern-Reconfigurable Antennas: \\Fundamental Models and Practical Algorithms
}

\author{Pinjun~Zheng, Yuchen~Zhang,~Tareq~Y.~Al-Naffouri, Md.~Jahangir Hossain, and Anas~Chaaban 

\thanks{
\violet{The officially published version is available at \url{https://ieeexplore.ieee.org/document/11612943}, but it contains a few minor formatting issues. The present version reflects the authors’ intended formatting.}

Pinjun Zheng, Md.~Jahangir Hossain, and Anas~Chaaban are with the School of Engineering, The University of British Columbia, Kelowna, BC V1V 1V7, Canada (e-mail: pinjun.zheng@ubc.ca; jahangir.hossain@ubc.ca; anas.chaaban@ubc.ca). 

Yuchen~Zhang and Tareq~Y.~Al-Naffouri are with the Electrical and Computer Engineering Program, Division of Computer, Electrical and Mathematical Sciences and Engineering (CEMSE), King Abdullah University of Science and Technology (KAUST), Thuwal, 23955-6900, Kingdom of Saudi Arabia (e-mail: yuchen.zhang@kaust.edu.sa; tareq.alnaffouri@kaust.edu.sa).

The work of P. Zheng is supported by funding from the Natural Sciences and Engineering Research Council of Canada under awards AWD-030231 NSERC 2024 and AWD-034729 NSERC 2026. The work of Y. Zhang and T. Y. AlNaffouri is supported by King Abdullah University of Science and Technology (KAUST) under Award No. ORFS-CRG12-2024-6478. (\textit{Corresponding author: Pinjun Zheng.})

} 
}

\maketitle

\begin{abstract}
The integration of pattern-reconfigurable antennas into hybrid multiple-input multiple-output (MIMO) architectures presents a promising path toward high-efficiency and low-cost transceiver solutions. Pattern-reconfigurable antennas can dynamically steer per-antenna radiation patterns, enabling more efficient power utilization and interference suppression. In this work, we study a tri-hybrid MIMO architecture for multi-user communications that integrates digital, analog, and antenna-domain precoding using pattern-reconfigurable antennas. For characterizing the reconfigurability of antenna radiation patterns, we develop two models---Model~I and Model~II. Model~I captures realistic hardware constraints through limited pattern selection, while Model~II explores the performance upper bound by assuming arbitrary pattern generation. Based on these models, we develop two corresponding tri-hybrid precoding algorithms grounded in the weighted minimum mean square error (WMMSE) framework, which alternately optimize the digital, analog, and antenna precoders under practical per-antenna power constraints. Realistic simulations conducted in ray-tracing generated environments are utilized to evaluate the proposed system and algorithms. The results demonstrate the significant potential of the considered tri-hybrid architecture in enhancing communication performance and hardware efficiency. However, they also reveal that the existing hardware is not yet capable of fully realizing these performance gains, underscoring the need for joint progress in antenna design and communication theory development.  
\end{abstract}

\section{Introduction}

Since their inception in the 1990s, \ac{MIMO} technologies have become a cornerstone of modern wireless communication systems~\cite{PAULRAJ2004Overview}. By deploying multiple antennas at both the transmitter and receiver, \ac{MIMO} systems leverage spatial diversity to mitigate fading and enhance link reliability, while employing spatial multiplexing to transmit independent data streams in parallel, thereby significantly improving spectral efficiency under constraints on bandwidth and power. These advantages have led to the widespread adoption of \ac{MIMO} in wireless standards such as LTE, 5G NR, and Wi-Fi~\cite{Li2010MIMO,Kim2015WLAN,Parkvall2017NR}. As communication systems shift toward higher-frequency bands, such as mmWave, large-scale antenna arrays are necessary to overcome increased path loss, which introduces challenges in hardware cost and system complexity. To address this problem, hybrid beamforming architecture has emerged as a promising solution, combining digital and analog processing to approximate the performance of fully digital systems while using a reduced number of \ac{RF} chains~\cite{Alkhateeb2014MIMO,Sohrabi2016Hybrid,Gao2016Energy}. With the ongoing evolution toward massive \ac{MIMO}~\cite{Marzetta2015Massive}, extra-large-scale massive \ac{MIMO}~\cite{Carvalho2020Stationarities}, and the growing emphasis on energy and cost efficiency in the future \ac{6G}~\cite{Andrews20246G}, the development of more efficient transceiver architectures remains a key research focus.

{Antenna-level reconfigurability is key to achieving efficient and resilient transceivers. Over the years, numerous reconfigurable antenna/surface designs have been explored for integration into \ac{MIMO} systems, including \acp{RIS}~\cite{Wang2024Wideband}, fluid antennas~\cite{Wong2021Fluid}, and more~\cite{Chen2025REMMA,Chen2025DBRAA,Shen2026Antenna}. \Acp{RIS} act as programmable reflectors that can intelligently reshape wireless channels, while fluid antennas (sometimes also referred to as movable antennas~\cite{Ning2025Movable,Li2026AISignalProcessing}) introduce additional degrees of freedom by reconfiguring the physical positions and/or orientations of antennas. Beyond these technologies, further reconfigurability is being realized through electromagnetic reconfigurable antennas, which enable dynamic adjustment of each antenna element's electromagnetic properties at the transmitter or receiver, such as radiation pattern, polarization, and operating frequency~\cite{Castellanos2025Embracing}. Additionally, novel designs like dynamic metasurface antennas employ waveguide-based leaky-wave structures to fundamentally alter the radiation mechanism~\cite{Williams2022Electromagnetic}, which have been demonstrated to exhibit high energy efficiency~\cite{Castellanos2025Embracing}.} 

While various reconfigurable antennas are being actively investigated, their integration into hybrid \ac{MIMO} systems remains insufficiently explored. In this work, we study a tri-hybrid \ac{MIMO} system~\cite{Heath2026Tri,Li2026Tri} that leverages pattern-reconfigurable antennas. By dynamically steering energy toward intended directions, these antennas enable more efficient power utilization and more effective interference suppression than their fixed-pattern counterparts~\cite{Zheng2025Enhanced}. This flexibility can be highly beneficial for enhancing conventional hybrid \ac{MIMO} systems. In fact, recent advances in the prototyping and experimental validation of pattern-reconfigurable antennas have initially demonstrated these advantages~\cite{Wang2025Electromagnetically}. Moreover, recent theoretical studies have begun to explore the potential of pattern-reconfigurable antennas in multi-user MIMO systems~\cite{Ying2024Reconfigurable}. For example,~\cite{Ying2025Reconfigurable} demonstrates notable throughput gains using pattern-reconfigurable antennas in fully digital arrays, while~\cite{Liu2025Tri} investigates tri-hybrid beamforming architectures that incorporate pattern reconfigurability. Although these works highlight the promise of such antennas, several challenges remain unaddressed. For instance,~\cite{Liu2025Tri} employs angular-domain discretization of radiation patterns, which leads to extremely large effective channel matrices and high computational complexity, particularly in 3D scenarios. More generally, existing models often either oversimplify hardware constraints or incur prohibitive computational costs, limiting their practical applicability.

To address these gaps, this paper presents a unified and physically grounded framework for tri-hybrid multi-user precoding with pattern-reconfigurable antennas. We aim to develop tractable models that accurately reflect physical operational principles, hardware limitations, and performance upper bounds. Such models enable the development of scalable tri-hybrid precoding algorithms that jointly optimize digital, analog, and antenna-domain processing under realistic constraints and with affordable computational complexity. Furthermore, we conduct realistic simulations using ray-tracing-based channels and state-of-the-art antenna prototypes, providing insights into the achievable gains with existing hardware. The main contributions of this paper are summarized as follows:

\begin{itemize} 

\item To establish a comprehensive foundation, we develop two models to characterize per-antenna radiation pattern reconfigurability within the tri-hybrid \ac{MIMO} architecture. The first model captures limited pattern selection, reflecting the practical operational principles of most existing reconfigurable antenna hardware. The second model permits arbitrary radiation pattern generation via spherical harmonics decomposition, representing an idealized upper bound that may become achievable with ongoing hardware advancements. Both models provide more tractable formulations compared to existing approaches, without introducing significant additional complexity. 

\item Using the developed models, we design two precoding algorithms based on the \ac{WMMSE} framework. These algorithms jointly optimize the digital, analog, and antenna-domain precoders in an alternating fashion, and incorporate practical \emph{per-antenna power constraints}. Both algorithms exhibit quadratic  (rather than exponential) complexity with respect to the number of antennas at \ac{BS}, ensuring their practicality in large-scale \ac{MIMO} systems.

\item We conduct extensive simulations to validate the considered tri-hybrid \ac{MIMO} system and the proposed algorithms. A ray-tracing-based channel generation is employed to emulate realistic propagation environments. Our simulations aim to answer the following key questions: 
\begin{itemize}
\item[\textbf{Q1}] How much performance gain can be achieved by incorporating pattern-reconfigurable antennas into multi-user hybrid \ac{MIMO} systems? {This question can be examined from two perspectives: (i) the ultimate potential of this technology, and (ii) the limits of current hardware.}
\item[\textbf{Q2}] Can the introduction of reconfigurable radiation patterns reduce the required number of costly \ac{RF} chains and antennas without compromising performance?
\end{itemize} 
\end{itemize}

The remainder of this paper is organized as follows. Section~\ref{sec:SACM} introduces the signal and channel models. Section~\ref{sec:IARPR} extends the channel model to incorporate antenna radiation pattern reconfigurability, leading to the proposed Model~I and Model~II. Based on these two models, we propose corresponding tri-hybrid precoding algorithms in Sections~\ref{sec:algoM1} and~\ref{sec:algoM2}, respectively. Simulation results are presented in Section~\ref{sec:sims}, and the conclusion is drawn in Section~\ref{sec:cons}.

We adopt the following notation throughout the paper. Non-bold italic lower and upper case letters (e.g., $a, A$) denote scalars, bold lower case letters (e.g., $\av$) denote vectors, and bold upper case letters (e.g., $\Am$) denote matrices. We use~$[\Am]_{i,j}$ to denote the entry at the~$i$$^\text{th}$ row and~$j$$^\text{th}$ column of the matrix~$\Am$. In addition,~$[\Am]_{n,:}$ and~$[\Am]_{:,m}$ denote $n$$^\text{th}$ row and $m$$^\text{th}$ column of~$\Am$, respectively. The superscripts ${(\cdot)}^\TT$, ${(\cdot)}^*$, ${(\cdot)}^\HH$, and ${(\cdot)}^{-1}$ represent the transpose, conjugate, Hermitian (conjugate transpose), and inverse operators, respectively. The notations $\|\cdot\|_0$, $\|\cdot\|_2$, and $\|\cdot\|_\mathsf{F}$ denote the $\ell_0$ quasi-norm, the $\ell_2$ norm, and the Frobenius norm, respectively. Additionally, {$\mathbb{R}$ denotes the set of all real numbers, {$\mathbb{R}_+$ denotes the set of all positive real numbers, $\mathbb{C}$ denotes the set of all complex numbers,} $\mathrm{Tr}(\cdot)$ denotes the trace of a square matrix, $\mathrm{Re}(\cdot)$ returns the real part, $\odot$ denotes the Hadamard (element-wise) product, and $\otimes$ represents the Kronecker product.

\section{Signal and Channel Models}\label{sec:SACM}
As illustrated in Fig.~\ref{fig_system}, we consider a downlink multi-user \ac{MIMO} system where an $N$-antenna \ac{BS} simultaneously serves $K$ users, each equipped with $M_k$ antennas. Unlike conventional digital-analog hybrid precoding, this work further incorporates antenna radiation pattern reconfigurability at the \ac{BS}, allowing each antenna's radiation pattern to be independently controlled and adjusted. Before analyzing the impact of radiation pattern reconfigurability, this section first outlines the adopted signal and channel models.

\begin{figure*}[t]
  \centering
  \includegraphics[width=0.8\linewidth]{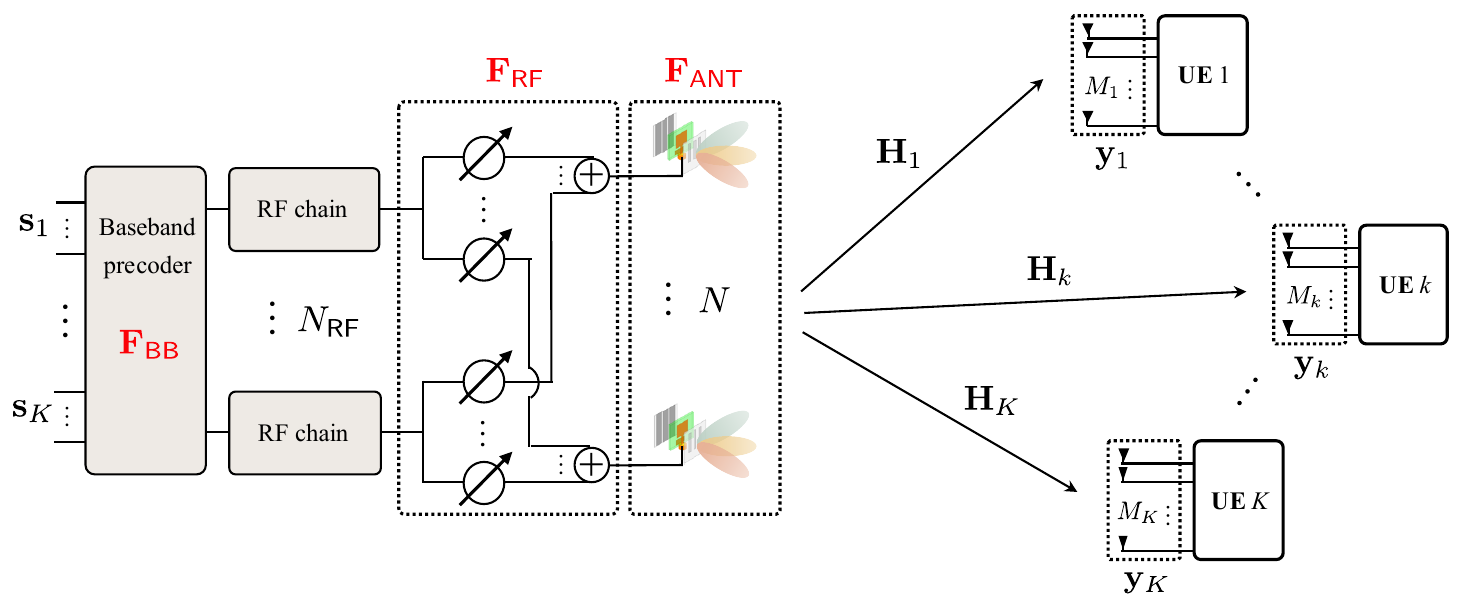}
  \caption{
Illustration of a multi-user \ac{MIMO} system employing a tri-hybrid precoding architecture at the BS. In addition to digital and analog precoding, the design also incorporates reconfigurable radiation patterns at each transmit antenna. The antenna-domain precoder, denoted as $\Fm_\mathsf{ANT}$, will be detailed in Section~\ref{sec:IARPR}. Depending on the adopted model, $\Fm_\mathsf{ANT}$ may represent either the selection matrix $\Fm_\mathsf{sel}$ or the coefficient matrix $\Fm_\mathsf{cof}$.}
  \label{fig_system}
\end{figure*}

\subsection{Signal Model}

Consider a \ac{BS} equipped with $N_\mathsf{RF}$ \ac{RF} chains and $N$ \ac{RF} phase shifters. Assume that the $k^\text{th}$ user requires $D_k$ data streams, and let $\sv_k\in\mathbb{C}^{D_k}$ denote this data stream vector. We concatenate all of these data vectors and define $\sv=[\sv_1^\TT,\sv_2^\TT,\dots,\sv_K^\TT]^\TT\in\mathbb{C}^{D}$, where $D=\sum_{k=1}^K D_k$ and we assume $\mathbb{E}[\sv\sv^\HH]=\mathbf{I}$. Let $\Fm_\mathsf{BB}\in\mathbb{C}^{N_\mathsf{RF}\times D}$ denote the baseband digital precoder and $\Fm_\mathsf{RF}\in\mathbb{C}^{N \times N_\mathsf{RF}}$ denote the \ac{RF} analog precoder. Then, the signals fed into the $N$ transmit antennas can be written as~\cite{Sohrabi2016Hybrid,Alkhateeb2015Limited}
\begin{equation}
	\xv = \Fm_\mathsf{RF}\Fm_\mathsf{BB}\sv=\sum_{k=1}^K\Fm_\mathsf{RF}\Fm_{\mathsf{BB},k}\sv_k,
\end{equation}
where $\Fm_{\mathsf{BB},k}\in\mathbb{C}^{N_\mathsf{RF}\times D_k}$ is the digital precoder corresponding to $\sv_k$ and $[\Fm_{\mathsf{BB},1},\Fm_{\mathsf{BB},2},\dots,\Fm_{\mathsf{BB},K}]=\Fm_\mathsf{BB}$. {Typically, the analog precoder $\mathbf{F}_{\mathsf{RF}}$ is physically implemented using phase shifters, which adjust the signal phase. This hardware nature leads to the constant-modulus constraint. In contrast, the digital precoder $\mathbf{F}_{\mathsf{BB}}$ operates in the baseband, where both amplitude and phase are fully controllable via digital signal processing, and is constrained solely by the transmit power~\cite{Ayach2014Spatially,Alkhateeb2015Limited}.} To be clear, this paper assumes that each entry of \( \Fm_\mathsf{RF} \) has a constant amplitude, satisfying $|[\Fm_\mathsf{RF}]_{i,j}|^2 = 1/N, \ \forall i,j$~\cite{Alkhateeb2015Limited,Ayach2014Spatially,Huang2024Hybrid}.
	Further, considering an independent power budget \( P_n \) for the power amplifier at each transmit antenna, this paper adopts a practical per-antenna power constraint, which can be written as~\cite{Yu2007Transmitter,Zhao2023Rethinking}  
\begin{equation}
	\big[\Fm_\mathsf{RF} \Fm_{\mathsf{BB}} \Fm_{\mathsf{BB}}^\HH \Fm_\mathsf{RF}^\HH\big]_{n,n}%
	\leq P_n,\ \forall n.
\end{equation}

Let $\Hm_k\in\mathbb{C}^{M_k\times N}$ denote the channel from the \ac{BS} to the $k^\text{th}$ user. The received signal at the $k^\text{th}$ user is expressed as
\begin{align}\label{eq:yk}
	\yv_k &= \Hm_k\xv+\nv_k,\notag\\
	&= \Hm_k \Fm_\mathsf{RF}\Fm_{\mathsf{BB},k}\sv_k + \sum_{i\neq k} \Hm_k\Fm_\mathsf{RF}\Fm_{\mathsf{BB},i}\sv_i + \nv_k,
\end{align}
where the first term is the desired signal for the $k^\text{th}$ user, the second term is the multi-user interference, and $\nv_k~\sim\mathcal{CN}(\mathbf{0},\sigma_k^2\mathbf{I})$ is the additive white Gaussian noise at the $k^\text{th}$ user.
Therefore, the spectral efficiency (rate) of the $k^\text{th}$ user is given by~\cite{Sohrabi2016Hybrid}
\begin{multline}\label{eq:Rk}
	R_k = \log_2 \det \bigg(\mathbf{I}_{M_k}+\Hm_k \Fm_\mathsf{RF}\Fm_{\mathsf{BB},k}\Fm_{\mathsf{BB},k}^\HH\Fm_\mathsf{RF}^\HH\Hm_k^\HH\\
	\times\Big(\sum_{i\neq k}\Hm_k\Fm_\mathsf{RF}\Fm_{\mathsf{BB},i}\Fm_{\mathsf{BB},i}^\HH\Fm_\mathsf{RF}^\HH\Hm_k^\HH+\sigma_k^2\mathbf{I}_{M_k}\Big)^{-1}\bigg).
\end{multline}
Subsequently, the weighted sum-rate of the entire system is
\begin{equation}\label{eq:WSR}
	R = \sum_{k=1}^K \beta_k R_k,
\end{equation}
where $\beta_k\in\mathbb{R}_+$ is a priority weight for the $k^\text{th}$ user. 

In general, the conventional hybrid precoding architecture jointly designs the baseband precoder $\Fm_\mathsf{BB}$ and the \ac{RF} precoder $\Fm_\mathsf{RF}$ to achieve desired signal reception. In this paper, we will demonstrate that the reconfigurability of the transmit antennas introduces additional flexibility in artificially shaping the wireless channels \( \{\Hm_k\}_{k=1}^K \). To ensure a thorough understanding, we begin by revisiting the fundamental definitions of antenna gain and radiation pattern, upon which the detailed expressions of the wireless channels will be built.

\subsection{Radiation Pattern and Antenna Gain}

\subsubsection{Radiation Pattern}\label{sec:RPD}
An \emph{antenna radiation pattern} is defined as a mathematical function or a graphical representation of the radiation properties of the antenna as a function of space coordinates. In most cases, the radiation pattern is determined in the far-field region of the radiator and is represented as a function of the directional coordinates (e.g., \ac{AoD} or \ac{AoA}). Radiation properties include power flux density, radiation intensity, field strength, directivity, phase, or polarization~\cite[Ch. 2]{Balanis2016Antenna}. A radiation pattern can be either a \emph{magnitude pattern} or a \emph{power pattern}. 

\subsubsection{Antenna Gain}
The \emph{antenna gain} is defined as the radiation intensity of an antenna in a given direction relative to that of an isotropic radiator~\cite{6758443,Balanis2016Antenna}. Typically, the antenna gain refers to the radiation \emph{power} intensity relative to an isotropic antenna. Consider a radiator with a total radiated power $P_\mathsf{rad}$. For an isotropic radiator, its radiation power intensity at any direction $(\theta,\phi)$ is given by $U_\mathsf{iso}(\theta,\phi) = P_\mathsf{rad}/{4\pi}$ (W/sr). Then, for a general radiator with a radiation power intensity distribution $U(\theta,\phi)$, its antenna gain is defined as 
\begin{equation}\label{eq:AGD}
	\mathsf{AG}(\theta,\phi)=\frac{U(\theta,\phi)}{U_\mathsf{iso}(\theta,\phi)} = \frac{4\pi U(\theta,\phi)}{P_\mathsf{rad}}.
\end{equation}
In decibels, we have $\mathsf{AG}_\mathsf{dBi}(\theta,\phi)=10\log_{10}\mathsf{AG}(\theta,\phi)$ (dBi).

In this paper, we define the \emph{magnitude gain} (or \emph{field gain}) of antennas as
\begin{equation}\label{eq:GD}
	G(\theta,\phi) = \sqrt{\mathsf{AG}(\theta,\phi)}.
\end{equation}
According to the definition of radiation pattern in Section~\ref{sec:RPD}, $G(\theta,\phi)$, as a function of spatial direction, is also a radiation pattern. In this paper, we do not distinguish them. Both antenna radiation pattern and antenna magnitude gain refer to $G(\theta,\phi)$. 

\begin{remark}\label{rmk:1}
	Considering the total radiation power $P_\mathsf{rad}$, we have~\cite{Balanis2016Antenna}
	\begin{equation}\label{eq:Uint}
		\int_0^{2\pi}\int_0^{\pi}U(\theta,\phi)\sin\theta \,\mathrm{d}\theta \,\mathrm{d}\phi = P_\mathsf{rad}.
	\end{equation}
Then, combining~\eqref{eq:AGD},~\eqref{eq:GD}, and~\eqref{eq:Uint} yields
\begin{align}\label{eq:4pi}
	\int_0^{2\pi}\int_0^{\pi}G^2(\theta,\phi)\sin\theta \,\mathrm{d}\theta \,\mathrm{d}\phi=4\pi.
\end{align}
\end{remark}
\begin{remark}\label{rmk:2}
	The antenna magnitude gain in linear scale must be strictly positive, i.e.,
	$
		G(\theta,\phi) > 0,\, \forall\, (\theta,\phi). 
	$ 
	It may take negative values when expressed in dBi.
\end{remark}

\subsection{A General MIMO Channel Model}
Based on the defined antenna magnitude gain, we present a general \ac{MIMO} channel model applicable to both far-field and near-field scenarios, which facilitates subsequent analysis. In a general narrowband multipath \ac{MIMO} system, the channel matrix \( \Hm_k \) in the frequency domain can be expressed as  
\begin{equation}\label{eq:genChan}
	\Hm_k = \sqrt{\frac{NM_k}{L_k}} \sum_{\ell=1}^{L_k} \Cm_{k,\ell} \odot \Am_{k,\ell} \odot \Gm_{k,\ell}^\mathsf{UE} \odot \Gm_{k,\ell}^\mathsf{BS},
\end{equation}
where \( L_k \) represents the total number of multipath components between the \ac{BS} and the \( k^\text{th} \) user,  
\( \Cm_{k,\ell} \in \mathbb{C}^{M_k \times N} \) stores the complex channel gains,  
\( \Am_{k,\ell} \in \mathbb{C}^{M_k \times N} \) captures the phase differences among array elements (i.e., the array manifold), and  
\( \Gm_{k,\ell}^\mathsf{UE}, \Gm_{k,\ell}^\mathsf{BS} \in \mathbb{R}_+^{M_k\times N} \) are the antenna magnitude gain matrices at the \( k^\text{th} \) user and the \ac{BS}, respectively. Specifically, for the per-antenna link from the \( n^\text{th} \) antenna at the \ac{BS} to the \( m^\text{th} \) antenna at the \( k^\text{th} \) user through the \( \ell^\text{th} \) path, the complex channel gain, phase difference relative to the reference antennas, transmit antenna magnitude gain, and receive antenna magnitude gain are given by \( [\Cm_{k,\ell}]_{m,n} \), \( [\Am_{k,\ell}]_{m,n} \), \( [\Gm_{k,\ell}^\mathsf{BS}]_{m,n} \), and \( [\Gm_{k,\ell}^\mathsf{UE}]_{m,n} \), respectively.

\subsubsection{Channel Component Expressions}
A set of commonly used expressions for these components is given by~\cite{Tarboush2021TeraMIMO}:
\begin{subequations}\label{eq:ChanPara}
\begin{align}
	[\Cm_{k,\ell}]_{m,n} &= \Big(\frac{\lambda}{4\pi d_{k,\ell}^{(mn)}}\Big)^{\zeta/2}e^{j\psi_{k,\ell}^{(mn)}},\\
	[\Am_{k,\ell}]_{m,n} &= \frac{1}{\sqrt{NM_k}} e^{-j\frac{2\pi}{\lambda}\big(d_{k,\ell}^{(mn)}-d_{k,\ell}\big)},\\
	\big[\Gm_{k,\ell}^\mathsf{BS}\big]_{m,n} &= G_{(n)}^\mathsf{BS}\big(\theta_{k,\ell}^{(mn)},\phi_{k,\ell}^{(mn)}\big),\label{eq:GBS}\\
	\big[\Gm_{k,\ell}^\mathsf{UE}\big]_{m,n} &= G_{k,(m)}^\mathsf{UE}\big(\vartheta_{k,\ell}^{(mn)},\varphi_{k,\ell}^{(mn)}\big),
\end{align}
\end{subequations}
where $m=1,2,\dots,M_k$ and $n=1,2,\dots,N$. 
Note that we use parentheses in the right-hand side expressions to highlight the antenna indices, which will be consistently used throughout the paper. Here, $\lambda$ is the signal wavelength, $d_{k,\ell}^{(mn)}$ denotes the propagation distance from the $n^\text{th}$ transmit antenna to the $m^\text{th}$ receive antenna through the path $(k,\ell)$, $\zeta$ is  the path loss exponent, $\psi_{k,\ell}^{(mn)}$ is a random phase shift, and $d_{k,\ell}$ is the propagation distance between two reference points at the \ac{BS} and the $k^\text{th}$ user. In addition, we denote each antenna's magnitude gain as $G_{(n)}^\mathsf{BS}$ or $G_{k,(m)}^\mathsf{UE}$, which is a function of signal \ac{AoD} or \ac{AoA}. Considering a 3D space, we use $\theta_{k,\ell}^{(mn)}$ and $\phi_{k,\ell}^{(mn)}$ to respectively denote the inclination and azimuth components of the \ac{AoD} at the BS, and use $\vartheta_{k,\ell}^{(mn)}$ and $\varphi_{k,\ell}^{(mn)}$ to denote that of the \ac{AoA} at the $k^{\text{th}}$ user.

Since~\eqref{eq:genChan} is a general model, our subsequent analysis is based on it without differentiating between far-field and near-field scenarios. This paper focuses on tri-hybrid precoding at the \ac{BS}, while assuming that all users are equipped with static antennas with fixed radiation patterns. Under this assumption, the reconfigurability of antenna radiation patterns can be captured by modeling the reconfigurable antenna gain in~\eqref{eq:GBS}. Following this, two models for pattern-reconfigurable antennas are introduced and discussed in the following Section~\ref{sec:IARPR}.
  
\section{Incorporating Antenna Radiation Pattern Reconfigurability}\label{sec:IARPR}

The reconfigurability of antenna radiation patterns can be modeled using two different approaches. The first is inspired by practical hardware designs, where antennas switch among a finite set of states, each corresponding to a distinct radiation pattern. A representative hardware prototype can be found in, e.g.,~\cite{Wang2025Electromagnetically}. The second approach considers an idealized antenna capable of generating arbitrary radiation patterns on demand, as studied in~\cite{Liu2025Tri}. This section presents a separate integration of both models into the \ac{MIMO} channel~\eqref{eq:genChan}. 

\subsection{Model I: Limited State Selection}

\begin{figure}[t]
  \centering
  \includegraphics[width=0.7\linewidth]{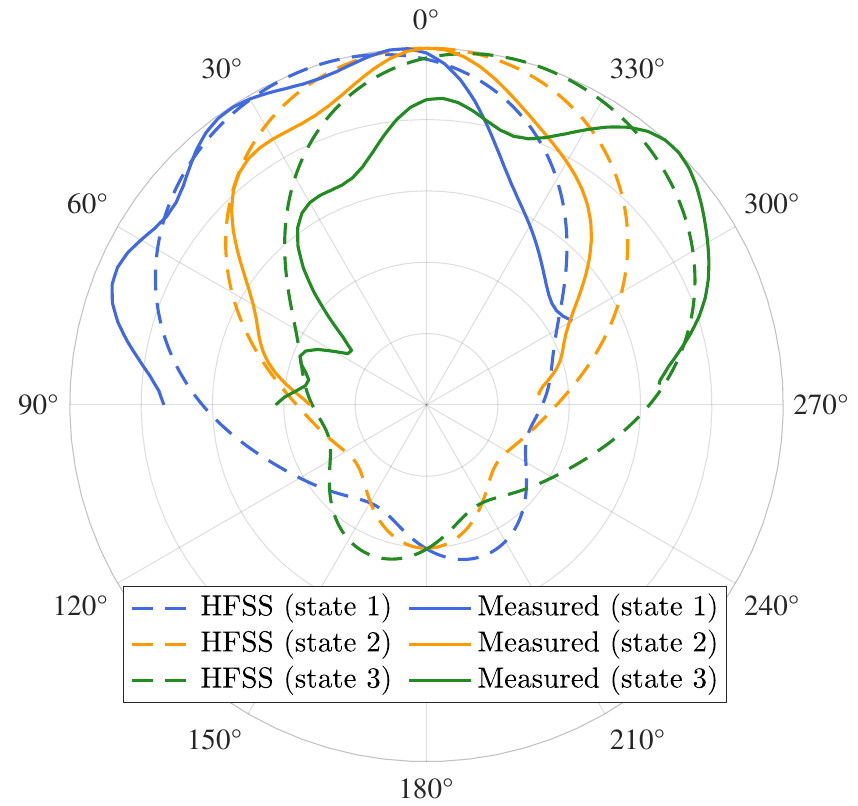}
  \caption{Illustration of three achievable radiation patterns of the reconfigurable antenna prototype in~\cite{Wang2025Electromagnetically}, showing both full-wave HFSS simulation and anechoic chamber measurement results.}
  \label{fig_trihybrid}
\end{figure}

The limited state selection model offers the most realistic representation, closely aligning with practical hardware designs. Figure~\ref{fig_trihybrid} illustrates an example of several achievable radiation patterns in a reconfigurable antenna prototype~\cite{Wang2025Electromagnetically}. In general, we can assume a set of $S$ available radiation pattern candidates, denoted as $\{\bar{G}_1, \bar{G}_2, \dots, \bar{G}_S\}$, from which each antenna can select its radiation pattern. Again, each $\bar{G}_s$ is a function of direction $(\theta,\phi)$. For a given angle $(\theta,\phi)$, we can define a candidate vector as $\bar{\gv}(\theta,\phi)=[\bar{G}_1(\theta,\phi),\bar{G}_2(\theta,\phi),\dots,\bar{G}_S(\theta,\phi)]^\TT\in\mathbb{R}_+^S$. Thus, the radiation pattern of the \ac{BS} antennas in~\eqref{eq:GBS} can be expressed as
\begin{equation}\label{eq:LSS}
	G_{(n)}^\mathsf{BS}(\theta,\phi) = \bar{\gv}^\TT(\theta,\phi)\bv_{(n)},
\end{equation}
where $\bv_{(n)}$ is a binary vector denoting the state selection of the $n^\text{th}$ antenna. Specifically, the following constraint applies:
\begin{equation}\label{eq:consbn}
\bv_{(n)} \in \left\{ \bv \in \{0, 1\}^S \mid \|\bv\|_0 = 1 \right\},\ \forall n.
\end{equation}

Now, we incorporate expression~\eqref{eq:LSS} into~\eqref{eq:genChan}. We first extend \ac{BS} antenna magnitude gain matrix $\Gm_{k,\ell}^{\mathsf{BS}}$ as 
\begin{equation}
	\bar{\Gm}_{k,\ell}^{\mathsf{BS}} \!=\!\! \begin{bmatrix}
		\bar{\gv}^\TT\big(\theta_{k,\ell}^{(11)},\phi_{k,\ell}^{(11)}\big) \!&\! \!\!\dots\!\! \!&\! \bar{\gv}^\TT\big(\theta_{k,\ell}^{(1N)},\phi_{k,\ell}^{(1N)}\big)\\
		\vdots \!&\!  \!\!\ddots\!\! \!&\! \vdots \\
		\bar{\gv}^\TT\!\big(\theta_{k,\ell}^{(M_k1)}\!\!,\phi_{k,\ell}^{(M_k1)}\big) \!&\! \!\!\dots\!\! \!&\! \bar{\gv}^\TT\!\big(\theta_{k,\ell}^{(M_kN)}\!\!,\phi_{k,\ell}^{(M_kN)}\big)
	\!\end{bmatrix}\!\!,
\end{equation}
and define a state selection matrix as
\begin{equation}\label{eq:Fsel_b}
	\Fm_{\mathsf{sel}} = \mathrm{blkdiag}\big\{\bv_{(1)},\bv_{(2)},\dots,\bv_{(N)}\big\},
\end{equation}
where $\mathrm{blkdiag}$ denotes the block diagonal operator that constructs a block diagonal matrix from its vector inputs.
Note that $\bar{\Gm}_{k,\ell}^{\mathsf{BS}}\in\mathbb{R}_+^{M_k\times NS}$ and $\Fm_{\mathsf{sel}}\in\{0,1\}^{NS\times N}$. Subsequently, channel model~\eqref{eq:genChan} can be rewritten as
\begin{align}\label{eq:HFM1}
	\Hm_k = \Hm_k^\mathsf{sel}\Fm_{\mathsf{sel}},
\end{align}
where~$\Hm_k^\mathsf{sel}$ is the effective channel defined as
\begin{multline}\label{eq:Heff1}
	\Hm_k^{\mathsf{sel}} \triangleq \sqrt{\frac{NM_k}{L_k}} \sum_{\ell=1}^{L_k} \Big(\big(\Cm_{k,\ell} \odot \Am_{k,\ell} \odot \Gm_{k,\ell}^\mathsf{UE}\big)\otimes \mathbf{1}_{1\times S}\Big)\\
	\odot \bar{\Gm}_{k,\ell}^{\mathsf{BS}}\in\mathbb{C}^{M_k\times NS},
\end{multline}

Substituting~\eqref{eq:HFM1} into~\eqref{eq:yk} yields the tri-hybrid signal model based on the limited state selection mechanism as:
\begin{equation}\notag
\boxed{
	\yv_k = \Hm_k^\mathsf{sel}\Fm_{\mathsf{sel}} \Fm_\mathsf{RF}\Fm_{\mathsf{BB},k}\sv_k + \sum_{i\neq k} \Hm_k^\mathsf{sel}\Fm_{\mathsf{sel}}\Fm_\mathsf{RF}\Fm_{\mathsf{BB},i}\sv_i + \nv_k.
}
\end{equation}
This model enables the joint adjustment of $\{\Fm_{\mathsf{sel}},\Fm_\mathsf{RF},\Fm_{\mathsf{BB}}\}$ to achieve the desired signal receptions.

\subsection{Model II: Arbitrary Radiation Pattern Generation}

\begin{figure}[t]
  \centering
  \includegraphics[width=\linewidth]{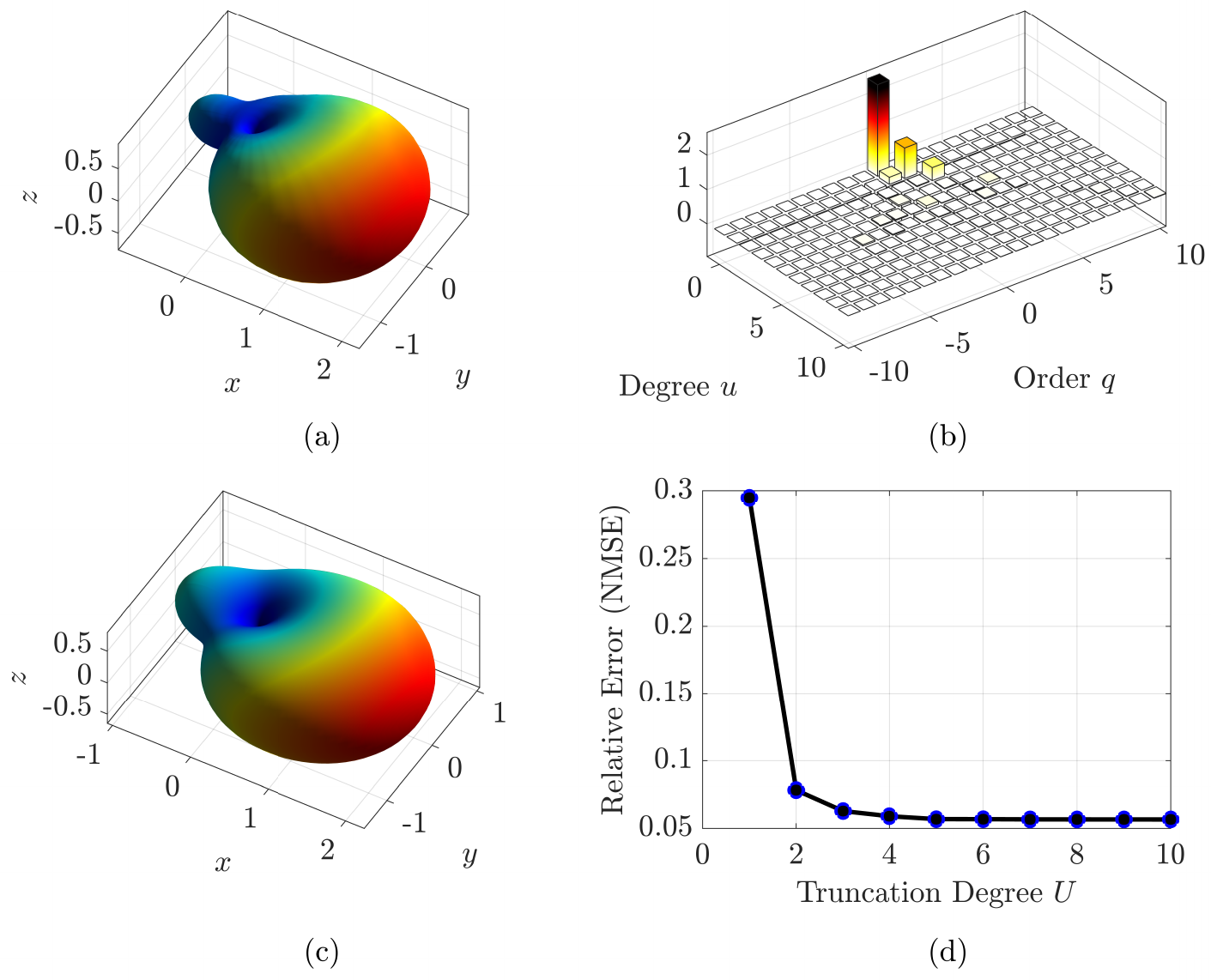}
  \vspace{-1.8em}
  \caption{
  A spherical harmonics decomposition example.
(a)~The radiation pattern for a selected state of the antenna design in ~\cite{Wang2025Electromagnetically}.
(b)~Visualization of its spherical harmonic coefficients, showing the dominance of low-degree coefficients, which justifies the truncation operation.
(c)~Truncated reconstruction with $U=4$. 
{(d)~Relative error between the truncated reconstruction and the original pattern as a function of truncation length $U$.} 
  }
  \label{fig_SHD}
\end{figure}

Apart from the limited state selection model, an alternative approach models pattern-reconfigurable antennas as capable of generating arbitrary radiation patterns. {We acknowledge that arbitrary pattern synthesis is idealized and may not be fully achievable with current hardware. However, this assumption serves to characterize the theoretical upper bound of system performance, providing valuable insight into the ultimate potential of reconfigurable antennas. Furthermore, ongoing advances in antenna design and fabrication are steadily closing the gap toward realizing such flexible pattern generation. For instance, preliminary studies on orthogonal antenna pattern generation have already been reported; see, e.g.,~\cite{Han2021Characteristic}.} 

A radiation pattern can be equivalently represented either in the angular domain~\cite{Liu2025Tri} or spherical domain~\cite{Ying2025Reconfigurable}. Typically, the angular domain representation results in a heavier effective channel matrix, whereas the spherical domain model is generally more compact and needs fewer coefficients due to the ``band-limited'' nature of most practical radiation patterns~\cite{Zheng2025Enhanced, Zheng2025Tri}. {For this reason, this paper adopts the spherical domain representation as in~\cite{Ying2025Reconfigurable}}. Based on spherical harmonics decomposition, any radiation pattern can be expressed as a superposition of an infinite series of spherical harmonics as
\begin{equation}\label{eq:SH}
    G_{(n)}^\mathsf{BS}(\theta,\phi)=\sum_{u=0}^{+\infty}\sum_{q=-u}^{u} c_{u q}^{(n)} Y_u^q(\theta,\phi).
\end{equation}
Here, $c_{u q}^{(n)}$ denotes the harmonic coefficient and $Y_u^q(\theta,\phi)$ is the \emph{real} spherical harmonic defined as~{\cite{schonefeld2005spherical}}
\begin{equation}\label{eq:Ylm}
Y_u^q(\theta,\phi)\!=\! \left\{
\begin{array}{ll}
\!\!\!\sqrt{2}N_u^qP_u^q(\cos{\theta})\cos{(q\phi)}, & q>0,\\
\!\!\!\sqrt{2}N_u^{|q|}P_u^{|q|}(\cos{\theta})\sin{(|q|\phi)}, & q< 0,\\
\!\!\!N_u^0P_u^0(\cos{\theta}), & q=0,
\end{array} \right.
\end{equation}
where $N_u^q=\sqrt{\frac{2u+1}{4\pi}\frac{(u-q)!}{(u+q)!}}$ is a normalization factor, and $P_u^q(\cos{\theta})$ represents the associated Legendre functions of $u^\text{th}$ degree and $q^\text{th}$ order. This set of functions $Y_{u}^q(\theta,\phi)$ constitutes a complete real orthonormal basis on the spherical space. 

To make the infinite series expansion in~\eqref{eq:SH} computationally manageable, we approximate it by truncating the series. This is motivated by the observation that most practical radiation patterns concentrate power in low-degree spherical harmonic components, as can be observed in Fig.~\ref{fig_SHD}-(b). Specifically, we retain terms in~\eqref{eq:SH} for $u=0,1,\dots,U$ only, which results in a total of $T=U^2+2U+1$ coefficients. Then, we approximate the radiation pattern of the $n^\text{th}$ antenna as
\begin{align}\label{eq:gn}
    G_{(n)}^\mathsf{BS}(\theta,\phi) \!\approx\! \sum_{u=0}^{U}\sum_{q=-u}^u c_{uq}^{(n)} Y_u^q(\theta,\phi) = \sum_{t=1}^T \tilde{c}_t^{(n)} \tilde{Y}_t(\theta,\phi),
\end{align}
where $\tilde{c}_t^{(n)} = c_{uq}^{(n)}$ and $\tilde{Y}_t(\theta,\phi) = Y_u^q(\theta,\phi)$, for $t=u^2+u+q+1$, $u\in[0,U]$, $q\in[-u,u]$. For convenience, we further concatenate $\cv_{(n)} \triangleq[\tilde{c}_1^{(n)},\tilde{c}_2^{(n)},\dots,\tilde{c}_T^{(n)}]^\TT\in\mathbb{R}^T$, 
$\gammav(\theta,\phi) \triangleq[\tilde{Y}_1(\theta,\phi),\tilde{Y}_2(\theta,\phi),\dots,\tilde{Y}_T(\theta,\phi)]^\TT\in\mathbb{R}^T$, where $\cv_{(n)}$ is the coefficient vector of the $n^\text{th}$ antenna element, and $\gammav(\theta,\phi)$ is the spherical basis vector. Then, we can rewrite~\eqref{eq:gn} as
\begin{equation}\label{eq:Gn_thetaphi}
    G_{(n)}^\mathsf{BS}(\theta,\phi) \approx \gammav^\TT(\theta,\phi)\cv_{(n)}.
\end{equation}
Based on~\eqref{eq:4pi} and the energy conservation law in spherical harmonics transform, we have the following constraint:
\begin{equation}\label{eq:conscn}
	\|\cv_{(n)}\|_2^2 = 4\pi,\ \forall n.
\end{equation}
In the spherical harmonics decomposition and the truncated reconstruction example illustrated in Fig.~\ref{fig_SHD}, we observe that taking $U\geq 4$ can generally achieve an accurate approximation. Additionally, we note that~\eqref{eq:Gn_thetaphi} has the same form as~\eqref{eq:LSS}. Actually, the $4\pi$ factor is also included in Model~I. Each candidate pattern in $\bar{\gv}(\theta,\phi)$, i.e., $\bar{G}_s(\theta,\phi)$, has a total power $\int_0^{2\pi}\int_0^{\pi}\bar{G}_s^2(\theta,\phi)\sin\theta d\theta d\phi = 4\pi$, according to Remark~\ref{rmk:1}. 

Similarly, we can define a harmonic coefficient matrix as
\begin{equation}\label{eq:Fcof_c}
	\Fm_{\mathsf{cof}} = \mathrm{blkdiag}\big\{\cv_{(1)},\cv_{(2)},\dots,\cv_{(N)}\big\},
\end{equation}
and an extended \ac{BS} antenna radiation pattern basis matrix as
\begin{equation}
	\Gammam_{k,\ell}^{\mathsf{BS}} \!=\!\! \begin{bmatrix}
		\gammav^\TT\big(\theta_{k,\ell}^{(11)},\phi_{k,\ell}^{(11)}\big) \!&\! \!\!\dots\!\! \!&\! \gammav^\TT\big(\theta_{k,\ell}^{(1N)},\phi_{k,\ell}^{(1N)}\big)\\
		\vdots \!&\!  \!\!\ddots\!\! \!&\! \vdots \\
		\gammav^\TT\!\big(\theta_{k,\ell}^{(M_k1)}\!\!,\phi_{k,\ell}^{(M_k1)}\big) \!&\! \!\!\dots\!\! \!&\! \gammav^\TT\!\big(\theta_{k,\ell}^{(M_kN)}\!\!,\phi_{k,\ell}^{(M_kN)}\big)
	\!\end{bmatrix}\!\!.
\end{equation}
Note that $\Fm_{\mathsf{cof}}\in\mathbb{R}^{NT\times N}$ and $\Gammam_{k,\ell}^{\mathsf{BS}}\in\mathbb{R}^{M_k\times NT}$.
Then, the effective channel $\Hm_k^{\mathsf{cof}}$ can be defined as 
\begin{multline}\label{eq:Heff2}
	\Hm_k^{\mathsf{cof}} \triangleq \sqrt{\frac{NM_k}{L_k}} \sum_{\ell=1}^{L_k} \Big(\big(\Cm_{k,\ell} \odot \Am_{k,\ell} \odot \Gm_{k,\ell}^\mathsf{UE}\big)\otimes \mathbf{1}_{1\times T}\Big)\\
	\odot \Gammam_{k,\ell}^{\mathsf{BS}}\in\mathbb{C}^{M_k\times NT}.
\end{multline}
Finally, we have a similar channel expression as 
\begin{equation}\label{eq:HFM2}
	\Hm_k = \Hm_k^\mathsf{cof}\Fm_\mathsf{cof}.
\end{equation}

Substituting~\eqref{eq:HFM2} into~\eqref{eq:yk} yields the other signal model as 
\begin{equation}\notag
\boxed{
	\yv_k = \Hm_k^\mathsf{cof}\Fm_{\mathsf{cof}} \Fm_\mathsf{RF}\Fm_{\mathsf{BB},k}\sv_k + \sum_{i\neq k} \Hm_k^\mathsf{cof}\Fm_{\mathsf{cof}}\Fm_\mathsf{RF}\Fm_{\mathsf{BB},i}\sv_i + \nv_k.
}
\end{equation}

\subsection{Comparison Between the Two Models}\label{sec:comparison}

Now, we have obtained two models, $\{\Hm_k^\mathsf{sel},\Fm_\mathsf{sel}\}$ and $\{\Hm_k^\mathsf{cof},\Fm_\mathsf{cof}\}$, to characterize the radiation pattern reconfigurability of the antennas at the \ac{BS}. The key differences between these models are summarized as follows:
\begin{itemize}
	\item \textbf{Different physical interpretations:} Model~I represents the antenna state selection, where \( \Fm_\mathsf{sel} \) denotes the state selection of the \ac{BS} antennas, and \( \Hm_k^\mathsf{sel} \) integrates all possible channels based on the available radiation pattern candidates. In contrast, Model~II models arbitrary radiation pattern synthesis, where \( \Fm_\mathsf{cof} \) represents the spherical harmonic coefficients, and \( \Hm_k^\mathsf{cof} \) is a concatenation of channels based on each spherical basis.
	\item \textbf{Different matrix dimensions:} Both models introduce a dimensional lift to the channels. The original channel is \( \Hm_k \in \mathbb{C}^{M_k \times N} \), the effective channel in Model~I is \( \Hm_k^\mathsf{sel} \in \mathbb{C}^{M_k \times NS} \), while the effective channel in Model~II is \( \Hm_k^\mathsf{cof} \in \mathbb{C}^{M_k \times NT} \). Here, \( S \) denotes the total number of available antenna states, and \( T \) represents the truncation length of the spherical harmonics decomposition.
	\item \textbf{Different constraints:} While \( \Fm_\mathsf{sel} \) and \( \Fm_\mathsf{cof} \) exhibit a similar block diagonal structure, as shown in \eqref{eq:Fsel_b} and \eqref{eq:Fcof_c}, they impose different constraints. Each diagonal block in \( \Fm_\mathsf{sel} \) is constrained as \( \bv_{(n)} \in \left\{ \bv \in \{0, 1\}^S \mid \|\bv\|_0 = 1 \right\} \), whereas in \( \Fm_\mathsf{cof} \), the constraint is \( \|\cv_{(n)}\|_2^2 = 4\pi \).
\end{itemize}

{
\begin{remark}
A critical concern is the potential increase in channel estimation overhead due to the enlarged effective channels \( \Hm_k^\mathsf{sel} \in \mathbb{C}^{M_k \times NS} \) and \( \Hm_k^\mathsf{cof} \in \mathbb{C}^{M_k \times NT} \), compared to \( \Hm_k \in \mathbb{C}^{M_k \times N} \). However, these two effective channels are not arbitrary matrices, but exhibit strong structures. In high-frequency bands, far-field channels are typically sparse in the angular domain and can be characterized by a limited number of physical parameters, such as the angles and gains of dominant propagation paths~\cite{Ayach2014Spatially,Tarboush2021TeraMIMO,Zheng2024Coverage}. Similarly, in the near-field scenario, the channels are highly structured in the spherical domain~\cite{Abdallah2025Exploring}. In both cases, the channel can be described by a parametric model determined by a limited set of parameters, from which the full channel matrix can be reconstructed based on the underlying structure. Importantly, the pattern-reconfigurability models in~\eqref{eq:LSS} and~\eqref{eq:Gn_thetaphi} use a structural expansion without introducing any new parameters. Hence, from a parametric channel estimation perspective, the overall estimation complexity does not increase. Furthermore, recent advances in deep learning-based channel estimation and beamforming suggest that explicit channel reconstruction may not be necessary, as latent channel features can be learned and directly leveraged for beamforming~\cite{Sohrabi2022Active,Zheng2025Model}, offering an alternative direction for addressing this issue.
\end{remark}
}

\section{Tri-Hybrid Precoding Based on Model~I}\label{sec:algoM1}

Based on Model~I, the considered tri-hybrid precoding problem aims to optimize \( \{\Fm_{\mathsf{sel}}, \Fm_\mathsf{RF}, \Fm_{\mathsf{BB}}\} \) jointly to maximize the weighted sum-rate defined in~\eqref{eq:WSR}. We assume that an estimate of the effective channels is available. For clarity, we redefine the notation as $\Hm_k^\mathsf{sel} = \bar{\Hm}_k^\mathsf{sel} + \Delta \Hm_k^\mathsf{sel},\ \forall k$, where $\Hm_k^\mathsf{sel}$ now denotes the estimated channel, $\bar{\Hm}_k^\mathsf{sel}$ is the ground truth, and $\Delta \Hm_k^\mathsf{sel}$ represents the estimation error. To simplify the optimization, we temporarily define a fully digital precoder \( \Fm_\mathsf{D} \!\triangleq\! \Fm_\mathsf{RF} \Fm_{\mathsf{BB}} \!\in\! \mathbb{C}^{N \times D} \), and seek the optimal solution of the reduced set \( \{\Fm_{\mathsf{sel}}, \Fm_\mathsf{D}\} \) instead. The optimized \( \Fm_\mathsf{D} \) is then decomposed back into \( \Fm_\mathsf{RF} \) and \( \Fm_{\mathsf{BB}} \). {While suboptimal, this two-stage method is widely adopted in the literature as it achieves a good balance between performance and tractability~\cite{Yu2016Alternating,Huang2024Hybrid}.}

Partitioning $\Fm_\Dt = [\Fm_{\Dt,1},\Fm_{\Dt,2},\dots,\Fm_{\Dt,K}]$, where $\Fm_{\Dt,k}\in\mathbb{C}^{N\times D_k}$, we formulate the initial precoding problem as: 
\begin{subequations}\label{eq:optM1}
	\begin{align}
		\max_{\Fm_{\mathsf{sel}}, \Fm_\mathsf{D}}&\ \sum_{k=1}^K\beta_k\log_2 \det \bigg(\mathbf{I}_{M_k}+\Hm_k^\mathsf{sel}\Fm_{\mathsf{sel}}\Fm_{\mathsf{D},k}\Fm_{\mathsf{D},k}^\HH\Fm_{\mathsf{sel}}^\HH(\Hm_k^\mathsf{sel})^\HH \notag \\
	&\hspace{-1em} \times\!\Big(\!\sum_{i\neq k}\Hm_k^\mathsf{sel}\Fm_{\mathsf{sel}}\Fm_{\mathsf{D},i}\Fm_{\mathsf{D},i}^\HH\Fm_{\mathsf{sel}}^\HH(\Hm_k^\mathsf{sel})^\HH\!+\!\sigma_k^2\mathbf{I}_{M_k}\!\Big)^{\!-1}\bigg)\\
	\mathrm{s.t.\ }&\ [\Fm_{\mathsf{D}}\Fm_{\mathsf{D}}^\HH]_{n,n}\leq P_n,\ \forall n, \\
	&\ \Fm_{\mathsf{sel}} = \mathrm{blkdiag}\big\{\bv_{(1)},\bv_{(2)},\dots,\bv_{(N)}\big\},\\
	&\ \bv_{(n)} \in \left\{ \bv \in \{0, 1\}^S \mid \|\bv\|_0 = 1 \right\},\ \forall n.
	\end{align}
\end{subequations}

\subsection{Conversion to the Weighted Sum MSE Minimization}
\label{sec:ConverWSMSEM}

According to~\cite{Shi2011Iteratively,Zhao2023Rethinking}, the formulated weighted sum-rate maximization problem is equivalent to a \ac{WMMSE} minimization problem. Specifically, by applying Lemma~1 in~\cite{Zhao2023Rethinking}, the maximization problem~\eqref{eq:optM1} can be recast as the following minimization problem:
\begin{subequations}\label{eq:WMMSE_M1}
	\begin{align}
		\min_{\Wm,\Um,\Fm_{\mathsf{sel}}, \Fm_\mathsf{D}}&\ \sum_{k=1}^K\beta_k \big(\mathrm{Tr}(\Wm_k\Em_k)-\ln\det(\Wm_k)\big)\\
	\mathrm{s.t.\quad\ }&\ [\Fm_{\mathsf{D}}\Fm_{\mathsf{D}}^\HH]_{n,n}\leq P_n,\ \forall n, \label{eq:PAPC}\\
	&\ \Fm_{\mathsf{sel}} = \mathrm{blkdiag}\big\{\bv_{(1)},\bv_{(2)},\dots,\bv_{(N)}\big\},\\
	&\ \bv_{(n)} \in \left\{ \bv \in \{0, 1\}^S \mid \|\bv\|_0 = 1 \right\},\ \forall n. \label{eq:PASSC}
	\end{align}
\end{subequations}
where 
\begin{multline}\label{eq:Ek}
	\Em_k \triangleq \big(\mathbf{I} - \Um_k^\HH\Hm_k^\mathsf{sel}\Fm_{\mathsf{sel}}\Fm_{\Dt,k}\big)\big(\mathbf{I} - \Um_k^\HH\Hm_k^\mathsf{sel}\Fm_{\mathsf{sel}}\Fm_{\Dt,k}\big)^\HH\\
	+ \Um_k^\HH\Big(\sum_{i\neq k} \Hm_k^\mathsf{sel}\Fm_{\mathsf{sel}}\Fm_{\Dt,i}\Fm_{\Dt,i}^\HH\Fm_{\mathsf{sel}}^\HH(\Hm_k^\mathsf{sel})^\HH+\sigma_k^2\mathbf{I}\Big)\Um_k.
\end{multline}
Here, \( \Wm=\{\Wm_k\in\mathbb{C}^{D_k\times D_k}\}_{k=1}^K \) and \( \Um=\{\Um_k\in\mathbb{C}^{M_k\times D_k}\}_{k=1}^K \) are two sets of auxiliary variables, with an additional constraint $\Wm_k\succ\mathbf{0},\ \forall k$.

Problem~\eqref{eq:WMMSE_M1} can be solved using the \ac{BCD} method~\cite{Bertsekas1997Nonlinear}. Although~\eqref{eq:WMMSE_M1} introduces additional optimization variables, the subproblems with respect to each auxiliary variable in \( \{\Wm, \Um\} \) are individually convex and admit closed-form solutions. As a result, the overall complexity is not significantly increased, while the objective function becomes more tractable compared to the original formulation in~\eqref{eq:optM1}. According to Lemma~1 in~\cite{Zhao2023Rethinking} and defining $\Hm_k=\Hm_k^\mathsf{sel}\Fm_{\mathsf{sel}}$, the optimal \( \{\Wm, \Um\} \) are given by
\begin{align}
	\Um_k^{\mathsf{opt}} &= \Big(\sum_{i=1}^K \Hm_k\Fm_{\Dt,i}\Fm_{\Dt,i}^\HH\Hm_k^\HH +\sigma_k^2\mathbf{I}\Big)^{-1}\Hm_k\Fm_{\Dt,k},\ \forall k,\label{eq:Uopt}\\
	\Wm_k^\mathsf{opt} &= \big(\mathbf{I}-(\Um_k^{\mathsf{opt}})^\HH\Hm_k\Fm_{\Dt,k} \big)^{-1},\ \forall k.\label{eq:Wopt}
\end{align}

\subsection{Per-Antenna Optimization of $\Fm_{\mathsf{sel}}$ and $\Fm_\mathsf{D}$} \label{sec:PAO_MI}

Now we focus on the optimization of $\{\Fm_{\mathsf{sel}}, \Fm_\mathsf{D}\}$ in~\eqref{eq:WMMSE_M1}. Considering the per-antenna power constraint~\eqref{eq:PAPC} and the per-antenna state selection constraint~\eqref{eq:PASSC}, we can naturally treat the variables related to each antenna as a block variable and update each antenna alternately~\cite{Zhao2023Rethinking,Armandoust2023MIMO}. To do so, we first conduct a partition to $\Fm_\Dt$ as follows:
\begin{align}\label{eq:FDexpres}
	\Fm_\Dt &= \begin{bmatrix}
		\Fm_{\Dt,1} & \dots & \Fm_{\Dt,k} & \dots & \Fm_{\Dt,K}
	\end{bmatrix}\in\mathbb{C}^{N\times D},\\
	&= \begin{bmatrix}
		\fv_{(1),1}^\HH & \dots & \fv_{(1),k}^\HH & \dots & \fv_{(1),K}^\HH \\
		\vdots & \ddots & \vdots & \ddots & \vdots \\
		\fv_{(N),1}^\HH & \dots & \fv_{(N),k}^\HH & \dots & \fv_{(N),K}^\HH
	\end{bmatrix}=\begin{bmatrix}
		\fv_{(1)}^\HH \\ \vdots \\ \fv_{(N)}^\HH
	\end{bmatrix},\notag
\end{align}
where $\fv_{(n),k}\in\mathbb{C}^{D_k}$ and $\fv_{(n)}\in\mathbb{C}^{D}$. Again $D=\sum_{k=1}^K D_k$. Then, we consider the optimization of \( \{\fv_{(n)}, \bv_{(n)}\} \) given \( \{\fv_{(i)}, \bv_{(i)}\}_{i \neq n} \), \( \Um_k \), and \( \Wm_k \). 

The per-antenna optimization problem is written as
\begin{align}\label{eq:optfb}
	\min_{\fv_{(n)}, \bv_{(n)}}&\ \sum_{k=1}^K\beta_k \big(\mathrm{Tr}(\Wm_k\Em_k)-\ln\det(\Wm_k)\big)\\
	\mathrm{s.t.\ \ }&\ \|\fv_{(n)}\|_2^2\leq P_n,\, \ \bv_{(n)} \in \left\{ \bv \in \{0, 1\}^S \mid \|\bv\|_0 = 1 \right\}. \notag
\end{align}
By further partitioning $\Hm_k^\mathsf{sel}\!=\![\Hm_{(1),k}^\mathsf{sel},\dots,\Hm_{(N),k}^\mathsf{sel}]\!\in\!\mathbb{C}^{M_k\!\times\! NS}$, where $\Hm_{(n),k}^\mathsf{sel}\in\mathbb{C}^{M_k\times S}$, we can express
\begin{equation}\label{eq:expressHFF}
	\Hm_k^\mathsf{sel}\Fm_{\mathsf{sel}}\Fm_{\Dt,i} = \sum_{n=1}^N \Hm_{(n),k}^\mathsf{sel}\bv_{(n)}\fv_{(n),i}^\HH,
\end{equation}
based on which the following Lemma~\ref{lemma:1} follows.
\begin{lemma}\label{lemma:1}
	Optimization problem~\eqref{eq:optfb} is equivalent to:
\begin{subequations}\label{eq:optfb2}
	\begin{align}
		\min_{\fv_{(n)}, \bv_{(n)}}&\ \|\fv_{(n)}\|_2^2\bv_{(n)}^\TT\Bm_{nn}\bv_{(n)} + 2\mathrm{Re}\Big(\fv_{(n)}^\HH(\Qm_n-\Dm_n)\bv_{(n)}\Big)\notag \\
	\mathrm{s.t.\ \ }&\ \|\fv_{(n)}\|_2^2\leq P_n,\\
	&\ \bv_{(n)} \in \left\{ \bv \in \{0, 1\}^S \mid \|\bv\|_0 = 1 \right\},\label{eq:bncons}
	\end{align}
\end{subequations}
where 
\begin{align}
\Bm_{qp} &= \sum_{k=1}^K\beta_k(\Hm_{(q),k}^\mathsf{sel})^\HH\Um_k\Wm_k\Um_k^\HH\Hm_{(p),k}^\mathsf{sel}\in\mathbb{C}^{S\times S},\label{eq:Bm}\\
\Qm_n &= \sum_{q\neq n} \fv_{(q)}\bv_{(q)}^\TT\Bm_{qn}\in\mathbb{C}^{D\times S},\label{eq:Qm} \\
\Dm_n &= 
	\begin{bmatrix}
		\beta_1\Wm_1\Um_1^\HH\Hm_{(n),1}^\mathsf{sel}\vspace{-0.5em}\\
		\vdots\\
		\beta_K\Wm_K\Um_K^\HH\Hm_{(n),K}^\mathsf{sel}
	\end{bmatrix}\in\mathbb{C}^{D\times S}.\label{eq:Dm}
\end{align}
\end{lemma}
\begin{proof}
	This can be proved by substituting~\eqref{eq:expressHFF} into~\eqref{eq:Ek} and~\eqref{eq:optfb} while omitting all constant terms.
\end{proof}

Considering the state selection constraint in~\eqref{eq:bncons}, the optimization variable \( \bv_{(n)} \) admits only \( S \) possible configurations. Therefore, an enumeration approach can be employed. For each candidate \( \bv_{(n)} \), the corresponding optimal \( \fv_{(n)} \) can be derived in closed form. Hence, the global optimum of~\eqref{eq:optfb2} can be obtained by evaluating and comparing all closed-form solutions of \( \fv_{(n)} \) across the \( S \) configurations of \( \bv_{(n)} \).
\begin{proposition}\label{prop:1}
	The optimal solution of~\eqref{eq:optfb} can be determined by applying an exhaustive search to $\bv_{(n)} \in \left\{ \bv \in \{0, 1\}^S \mid \|\bv\|_0 = 1 \right\}$. Given a candidate $\bv_{(n)}$, the corresponding optimal $\fv_{(n)}$ is determined as
	\begin{multline}\label{eq:fn_opt}
		\fv_{(n)}^\mathsf{opt} = (\Dm_n-\Qm_n)\bv_{(n)}\\ \times\min\Big(\frac{1}{\bv_{(n)}^\TT\Bm_{nn}\bv_{(n)}},\frac{\sqrt{P_n}}{\|(\Qm_n-\Dm_n)\bv_{(n)}\|_2}\Big).
	\end{multline}
\end{proposition}
\begin{proof}
	See Appendix.
\end{proof}
{
\begin{remark}
	Since the adopted per-antenna optimization framework allows decoupling of the optimization of the precoders related to different antennas, the enumeration approach is applied to each $\bv_{(n)}$ independently. Hence, there are only $S$ possible configurations to consider at each search step. In practice, the number of states per reconfigurable antenna is typically modest (e.g., 64 in~\cite{Wang2025Electromagnetically}). Moreover, for each candidate antenna state, the optimal precoder vector is obtained in closed form~\eqref{eq:fn_opt}.  These factors make the overall computational complexity affordable. 
\end{remark}
}

\subsection{Decomposition of $\Fm_\Dt$ into $\Fm_{\mathsf{RF}}$ and $\Fm_{\mathsf{BB}}$}\label{sec:decomp}
Given \( \{ \fv_{(n)}^\mathsf{opt}, \bv_{(n)}^\mathsf{opt} \}_{n=1}^N \), we can recover $\Fm_\mathsf{sel}^\mathsf{opt}$ and $\Fm_\Dt^\mathsf{opt}$ according to~\eqref{eq:Fsel_b} and~\eqref{eq:FDexpres}, respectively. A final step is to decompose $\Fm_\Dt^\mathsf{opt}$ back into $\Fm_{\mathsf{RF}}^\mathsf{opt}$ and $\Fm_{\mathsf{BB}}^\mathsf{opt}$ by solving the following minimization problem:
\begin{subequations}\label{eq:FDdecomp}
	\begin{align}
	\min_{\Fm_{\mathsf{RF}},\Fm_{\mathsf{BB}}}&\ \|\Fm_\Dt^\mathsf{opt}-\Fm_\mathsf{RF}\Fm_\mathsf{BB}\|_\mathsf{F}^2\\
	\mathrm{s.t.\ \ }&\ |[\Fm_\mathsf{RF}]_{i,j}|^2 = 1/N, \ \forall i,j,\\
	&\  \big[\Fm_\mathsf{RF} \Fm_{\mathsf{BB}} \Fm_{\mathsf{BB}}^\HH \Fm_\mathsf{RF}^\HH\big]_{n,n} \leq P_n, \forall n.\label{eq:PAPC2}
\end{align}
\end{subequations}
Regardless of the per-antenna power constraints~\eqref{eq:PAPC2}, various methods have been proposed to address this decomposition problem (see, e.g.,~\cite{Yu2016Alternating,Huang2024Hybrid}). Based on this, a straightforward approach to solve~\eqref{eq:FDdecomp} is to first relax the per-antenna power constraint~\eqref{eq:PAPC2} to a total power constraint of the form \( \mathrm{Tr}\big(\Fm_\mathsf{RF} \Fm_\mathsf{BB} \Fm_\mathsf{BB}^\HH \Fm_\mathsf{RF}^\HH\big) \leq \mathrm{Tr}\big(\Fm_\mathsf{D}^\mathsf{opt} (\Fm_\mathsf{D}^\mathsf{opt})^\HH\big) \). Then, an existing decomposition algorithm (such as Algorithm~1 in~\cite{Huang2024Hybrid}) can be applied to obtain an initial solution \( \tilde{\Fm}_\mathsf{RF}^\mathsf{opt} \) and \( \tilde{\Fm}_\mathsf{BB}^\mathsf{opt} \). Next, to enforce the per-antenna power constraint~\eqref{eq:PAPC2}, a normalization step can be performed on the baseband precoder as follows:
\begin{equation}
	{\Fm}_\mathsf{BB}^\mathsf{opt} \!=\! \tilde{\Fm}_\mathsf{BB}^\mathsf{opt}\min_n\Bigg(1,\sqrt{\frac{P_n}{\big[\tilde{\Fm}_\mathsf{RF}^\mathsf{opt} \tilde{\Fm}_{\mathsf{BB}}^\mathsf{opt} (\tilde{\Fm}_{\mathsf{BB}}^\mathsf{opt})^\HH (\tilde{\Fm}_\mathsf{RF}^\mathsf{opt})^\HH\big]_{n,n}}}\Bigg)\!,
\end{equation}
while keeping \( \Fm_\mathsf{RF}^\mathsf{opt} = \tilde{\Fm}_\mathsf{RF}^\mathsf{opt} \) unchanged.

\subsection{Algorithm Summary and Complexity Analysis}

The proposed tri-hybrid precoding algorithm based on Model~I is summarized in Algorithm~\ref{algo:MI}. The computational complexity of Algorithm~\ref{algo:MI} is analyzed as follows. We consider a massive \ac{MIMO} system where \( N \gg \sum_{k=1}^K M_k > D > K \), and thus focus on the complexity with respect to the number of \ac{BS} antennas \(N\) and the number of available antenna states~\(S\). The complexity of lines~4--6 is dominated by line~4, which has a complexity of \( \mathcal{O}(N^2S) \). The operation in line~8 executed over the loop in line~7 incurs a total complexity of \( \mathcal{O}(NS^2) \), while the updates in lines~9--12 over the loop in line~7 contribute a total complexity of \( \mathcal{O}(NS^3) \). Additionally, the decomposition method in line~17 has a complexity of \( \mathcal{O}(N_\mathsf{RF}^2N + N_\mathsf{RF}ND) \) when Algorithm~1 in~\cite{Huang2024Hybrid} is adopted. Since \( N_\mathsf{RF} > D \), the overall complexity of Algorithm~\ref{algo:MI} can be summarized as $\mathcal{O}\big(I_\mathsf{max}(N^2S + NS^3 + N_\mathsf{RF}^2N)\big)$, where \( I_\mathsf{max} \) denotes the maximum number of iterations.

\begin{algorithm}[t]
 \caption{Tri-Hybrid Precoding Based on Model~I}
 \label{algo:MI}
 \begin{algorithmic}[1]
 \State \textbf{Input:} $\Hm_k^\mathsf{sel},\sigma_k^2,\ \forall k$.\qquad \textbf{Output:} $\Fm_\mathsf{sel}^\mathsf{opt},\Fm_\mathsf{RF}^\mathsf{opt},\Fm_\mathsf{BB}^\mathsf{opt}$.
 \State Initialize $\Fm_\mathsf{sel},\Fm_\mathsf{RF},\Fm_\mathsf{BB}$, $\Fm_\Dt\!=\!\Fm_\mathsf{RF}\Fm_\mathsf{BB}$, $\Wm_k\!=\!\mathbf{I}$, {$f^{(0)}=\infty$,  convergence threshold $\epsilon$, and maximum iteration $I_\mathsf{max}$.}
 \For{$I=1,2,\dots,I_\mathsf{max}$}
 \State Compute $\Hm_k=\Hm_k^\mathsf{sel}\Fm_{\mathsf{sel}},\ \forall k$.
 \State Update $\Um_k=\Um_k^\mathsf{opt},\ \forall k$, according to~\eqref{eq:Uopt}.
 \State Update $\Wm_k=\Wm_k^\mathsf{opt},\ \forall k$, according to~\eqref{eq:Wopt}.
 \For{$n=1,2,\dots,N$}
 \State Compute $\{\Bm_{nn},\Qm_{n},\Dm_{n}\}$ according to~\eqref{eq:Bm}--\eqref{eq:Dm}.
 \For{$s=1,2,\dots,S$}
 \State Let $\bv_{(n)}=\mathbf{0}_{S\times 1}$ and assign $[\bv_{(n)}]_s =1$.
 \State Compute the optimal $\fv_{(n)}^\mathsf{opt}$ given $\bv_{(n)}$ via~\eqref{eq:fn_opt}.
 \EndFor
 \State Determine the optimal $\bv_{(n)}^\mathsf{opt}$ and the corresponding $\fv_{(n)}^\mathsf{opt}$ that minimize the objective function in~\eqref{eq:optfb2}.
 \EndFor
 \State Recover $\Fm_\mathsf{sel}^\mathsf{opt}$ and $\Fm_\Dt^\mathsf{opt}$ based on \( \big\{ \fv_{(n)}^\mathsf{opt}, \bv_{(n)}^\mathsf{opt} \big\}_{n=1}^N \) using~\eqref{eq:Fsel_b} and~\eqref{eq:FDexpres}, respectively.
 \State {$f^{(I)}=\sum_{k=1}^K\beta_k \big(\mathrm{Tr}(\Wm_k\Em_k)-\ln\det(\Wm_k)\big)$.}
 \State {If $I>1$ and $\big|f^{(I)}-f^{(I-1)}\big|<\epsilon$, break.}
 \EndFor 
 \State Decompose $\Fm_\Dt^\mathsf{opt}$ back into $\Fm_{\mathsf{RF}}^\mathsf{opt}$ and $\Fm_{\mathsf{BB}}^\mathsf{opt}$ using the method described in Section~\ref{sec:decomp}.
 \end{algorithmic} 
\end{algorithm}

\section{Tri-Hybrid Precoding Based on Model~II}\label{sec:algoM2}

The tri-hybrid precoding based on Model~II can be addressed by slightly modifying the approach proposed in Section~\ref{sec:algoM1}. Specifically, we continue to apply the \ac{WMMSE} principle to reformulate the optimization problem {following the same steps as in Section~\ref{sec:ConverWSMSEM}}, but with a different constraint, as follows: 
\begin{subequations}\label{eq:WMMSE_M2}
	\begin{align}
		\min_{\Wm,\Um,\Fm_{\mathsf{cof}}, \Fm_\mathsf{D}}&\ \sum_{k=1}^K\beta_k \big(\mathrm{Tr}(\Wm_k\Em_k)-\ln\det(\Wm_k)\big)\label{eq:M2obj}\\
	\mathrm{s.t.\quad\ }&\ [\Fm_{\mathsf{D}}\Fm_{\mathsf{D}}^\HH]_{n,n}\leq P_n,\ \forall n, \\
	&\ \Fm_{\mathsf{cof}} = \mathrm{blkdiag}\big\{\cv_{(1)},\cv_{(2)},\dots,\cv_{(N)}\big\},\\
	&\ \|\cv_{(n)}\|_2^2 = 4\pi,\ \forall n. \label{eq:coeCons}
	\end{align}
\end{subequations}
where 
\begin{multline}\label{eq:Ek2}
	\Em_k \triangleq \big(\mathbf{I} - \Um_k^\HH\Hm_k^\mathsf{cof}\Fm_\mathsf{cof}\Fm_{\Dt,k}\big)\big(\mathbf{I} - \Um_k^\HH\Hm_k^\mathsf{cof}\Fm_\mathsf{cof}\Fm_{\Dt,k}\big)^\HH\\
	+ \Um_k^\HH\Big(\sum_{i\neq k} \Hm_k^\mathsf{cof}\Fm_\mathsf{cof}\Fm_{\Dt,i}\Fm_{\Dt,i}^\HH\Fm_{\mathsf{cof}}^\HH(\Hm_k^\mathsf{cof})^\HH+\sigma_k^2\mathbf{I}\Big)\Um_k.
\end{multline}
This optimization problem can still be solved using the \ac{BCD} method, where the optimal solution of \( \{\Wm, \Um\} \) is obtained from~\eqref{eq:Uopt} and~\eqref{eq:Wopt}, given that \( \Hm_k = \Hm_k^\mathsf{cof} \Fm_\mathsf{cof} \). The key difference lies in the per-antenna optimization of \( \Fm_\mathsf{cof} \) and \( \Fm_\Dt \), as outlined in the following subsection.

\subsection{Per-Antenna Optimization of $\Fm_{\mathsf{cof}}$ and $\Fm_\mathsf{D}$}

Similar to Section~\ref{sec:PAO_MI}, considering the per-antenna power constraint and the per-antenna pattern synthesis, we can naturally treat the variables related to each antenna as a block variable and update each antenna alternately. 
Following this per-antenna optimization framework, we now focus on optimizing \( \fv_{(n)},\cv_{(n)}\) to maximize~\eqref{eq:M2obj}:
\begin{equation}\label{eq:optfbM2}
\begin{aligned}
	\min_{\fv_{(n)}, \cv_{(n)}}&\ \sum_{k=1}^K\beta_k \big(\mathrm{Tr}(\Wm_k\Em_k)-\ln\det(\Wm_k)\big) \\
	\mathrm{s.t.\ \ }&\ \|\fv_{(n)}\|_2^2\leq P_n,\quad \|\cv_{(n)}\|_2^2 = 4\pi.
\end{aligned}
\end{equation}
Given the expression $\Hm_k^\mathsf{cof}\Fm_{\mathsf{cof}}\Fm_{\Dt,i} = \sum_{n=1}^N \Hm_{(n),k}^\mathsf{cof}\cv_{(n)}\fv_{(n),i}^\HH$, we can apply Lemma~\ref{lemma:1} and recast the per-antenna optimization as the following problem:
\begin{align}
	\min_{\fv_{(n)}, \cv_{(n)}}&\ \|\fv_{(n)}\|_2^2\cv_{(n)}^\TT\Bm_{nn}\cv_{(n)} + 2\mathrm{Re}\Big(\fv_{(n)}^\HH(\Qm_n-\Dm_n)\cv_{(n)}\Big) \notag \\
	\mathrm{s.t.\ \ }&\ \|\fv_{(n)}\|_2^2\leq P_n,\quad \|\cv_{(n)}\|_2^2 = 4\pi, \label{eq:optfb2_M2}
\end{align}
where $\Bm_{nn}$, $\Qm_n$, and $\Dm_n$ can be computed via~\eqref{eq:Bm}--\eqref{eq:Dm} by replacing $\{\Hm_{(n),k}^\mathsf{sel},\bv_{(n)}\}$ with $\{\Hm_{(n),k}^\mathsf{cof},\cv_{(n)}\}$. Here, $\Hm_{(n),k}^\mathsf{cof}\in\mathbb{C}^{M_k\times T}$ is defined as a submatrix of $\Hm_{k}^\mathsf{cof}$ such that $\Hm_k^\mathsf{cof}\!=\![\Hm_{(1),k}^\mathsf{cof},\dots,\Hm_{(N),k}^\mathsf{cof}]\!\in\!\mathbb{C}^{M_k \times NT}$. Since $\cv_{(n)} \in \mathbb{R}^T$ has infinitely many possible configurations, exhaustive enumeration is not applicable. Therefore, we adopt the \ac{BCD} method again and alternately optimize $\fv_{(n)}$ and $\cv_{(n)}$. 

\subsubsection{Optimizing $\fv_{(n)}$ given $\cv_{(n)}$}
Using Proposition~\ref{prop:1}, the optimal $\fv_{(n)}$ for a fixed $\cv_{(n)}$ is given by:
\begin{multline}\label{eq:fn_opt2}
	\fv_{(n)}^\mathsf{opt} = 
	(\Dm_n - \Qm_n)\cv_{(n)}\\ \times\min\bigg(\frac{1}{\cv_{(n)}^\TT \Bm_{nn} \cv_{(n)}}, \frac{\sqrt{P_n}}{\|(\Qm_n - \Dm_n)\cv_{(n)}\|_2}\bigg).
\end{multline}

\subsubsection{Optimizing $\cv_{(n)}$ given $\fv_{(n)}$}\label{sec:rho}
However, the optimization of $\cv_{(n)}$ given a fixed $\fv_{(n)}$, formulated as:
\begin{align}
	\min_{\cv_{(n)}}\quad & \|\fv_{(n)}\|_2^2 \cv_{(n)}^\TT \Bm_{nn} \cv_{(n)} + 2\,\mathrm{Re}\left\{ \fv_{(n)}^\HH (\Qm_n - \Dm_n)\cv_{(n)} \right\} \notag \\
		\text{s.t.} \quad & \|\cv_{(n)}\|_2^2 = 4\pi, \label{eq:MIIC3}
\end{align}
does not admit a closed-form solution.\footnote{Although the global optimum of~\eqref{eq:MIIC3} does not have a closed-form expression,~\cite{Zheng2025Tri} shows that closed-form solutions exist for two local optima.} 
Before delving into the solution of~\eqref{eq:MIIC3}, we reveal an implicit constraint involved. 

Recalling~\eqref{eq:Gn_thetaphi}, we note that there is an implicit constraint on $\cv_{(n)}$ that the synthesized radiation pattern needs to be strictly positive, as demonstrated in Remark~\ref{rmk:2}. Since~\eqref{eq:MIIC3} does not include this constraint, it potentially leads to an optimized radiation pattern with non-positive values, which has no physical feasibility. By further examining~\eqref{eq:Gn_thetaphi}, we observe that the first entry in $\gammav(\theta,\phi)$ is the $0^\text{th}$-order spherical harmonic $Y_0^0(\theta,\phi)$, which is a sphere with a constant positive value across all $(\theta,\phi)$. Therefore, under the power constraint in~\eqref{eq:MIIC3}, we can always assign a sufficiently large coefficient to $Y_0^0(\theta,\phi)$ to ensure a strictly positive radiation pattern throughout the optimization process. To do so, we fix the first entry in $\cv_{(n)}$ as $2\sqrt{\rho\pi}$ and partition 
\begin{equation}\label{eq:ImpliConst}
	\cv_{(n)}=\begin{bmatrix} 2\sqrt{\rho\pi} & 2\sqrt{(1-\rho)\pi}\tilde{\cv}_{(n)}^\TT \end{bmatrix}^\TT\in\mathbb{R}^{T},
\end{equation} 
where then $\tilde{\cv}_{(n)}\in\mathbb{R}^{T-1}$ is the variable we optimize with the power constraint converted as $\|\tilde{\cv}_{(n)}\|_2^2=1$. Here, $\rho\in(0,1]$ is a factor indicating the strength of the isotropic component in the synthesized radiation pattern.\footnote{{Empirically, maintaining $\rho \geq 0.7$ in general ensures that the synthesized radiation pattern remains nonnegative over the sphere, as will be illustrated in Fig.~\ref{fig_sim09}. We also acknowledge that this constraint on the isotropic harmonic is not tight in characterizing the nonnegativity of the radiation pattern. A more elegant formulation that directly models and enforces the pattern nonnegativity may exist, which we leave for future investigation.}} A smaller value of $\rho$ stands for a higher degree of freedom afforded to $\cv_{(n)}$.

By substituting~\eqref{eq:ImpliConst} into \eqref{eq:MIIC3} and omitting constant terms, one can further recast~\eqref{eq:MIIC3} as:
\begin{subequations}\label{eq:MIIopt4}
	\begin{align}
	\min_{\tilde{\cv}_{(n)}}\quad & \tilde{\cv}_{(n)}^\TT\tilde{\Bm}_{nn} \tilde{\cv}_{(n)} + (\vv_1+\vv_2)^\TT\tilde{\cv}_{(n)}  \\
		\text{s.t.} \quad & \|\tilde{\cv}_{(n)}\|_2^2 = 1, \label{eq:MIIC4}
\end{align}
\end{subequations}
where $\tilde{\Bm}_{nn}\!=\!4\pi(1\!-\!\rho)\|\fv_{(n)}\|_2^2[\Bm_{nn}]_{2:T,2:T}$, $\vv_1\!=\!4\sqrt{(1\!-\!\rho)\pi}\mathrm{Re}\{\fv_{(n)}^\HH[\Qm_n\!-\!\Dm_n]_{:,2:T}\}$, and~$\vv_2 \!=\! 8\pi\sqrt{\rho(1\!-\!\rho)}\|\fv_{(n)}\|_2^2\mathrm{Re}\{[\Bm_{nn}]_{2:T,1}\}$. We employ the Riemannian manifold optimization technique to solve~\eqref{eq:MIIopt4}. The detailed steps are presented in the next subsection.

\subsection{Solving~\eqref{eq:MIIopt4} Using Riemannian Manifold Optimization}\label{sec:RieOpt}

Riemannian manifold optimization is a method for solving constrained problems by exploiting the geometry of the constraint set. Instead of operating in the full Euclidean space, it performs optimization directly on smooth manifolds where the constraints naturally define the search space~\cite{Boumal2023Introduction,Douik2020Precise,Zheng2024LEO}. In our case, the constraint~\eqref{eq:MIIC4} corresponds to a unit 2-norm sphere. We define the unit 2-norm manifold as 
\begin{equation}
	\mathbb{M} \triangleq \big\{\xv\in\mathbb{R}^{T-1}\ \big|\ \|\xv\|_2^2=1 \big\}.
\end{equation}
Next, we can rewrite~\eqref{eq:MIIopt4} in an unconstrained form as
\begin{align}\label{eq:RieOpt}
	\min_{\tilde{\cv}_{(n)}\in\mathbb{M}}\quad & \tilde{\cv}_{(n)}^\TT\tilde{\Bm}_{nn} \tilde{\cv}_{(n)} + (\vv_1+\vv_2)^\TT\tilde{\cv}_{(n)}.
\end{align}

Similar to the gradient-based algorithms in Euclidean space, optimization over a Riemannian manifold is implemented by using the \emph{Riemannian gradient}~\cite{Absil2009Optimization,Boumal2023Introduction,manopt}. We denote the objective function in~\eqref{eq:RieOpt} as $g\big(\tilde{\cv}_{(n)}\big)$. At iteration $j$, given the previous candidate variable~$\tilde{\cv}_{(n)}^{[j-1]}$, we first compute the Euclidean gradient as 
\begin{equation}
	\nabla_{\tilde{\cv}_{(n)}^{[j-1]}} g = (\tilde{\Bm}_{nn}+\tilde{\Bm}_{nn}^\TT)\tilde{\cv}_{(n)}^{[j-1]} +  \vv_1+\vv_2.
\end{equation}
Next, the Riemannian gradient is obtained by projecting the Euclidean gradient onto the tangent space at point $\tilde{\cv}_{(n)}^{[j-1]}$ of the manifold $\mathbb{M}$. Mathematically, the tangent space at a point $\tilde{\cv}_{(n)}\in\mathbb{M}$ is defined as $\mathcal{T}_{\tilde{\cv}_{(n)}}\mathbb{M} \triangleq\big\{\xv\in\mathbb{R}^{T-1}\ |\ \xv^\TT\tilde{\cv}_{(n)} = 0 \big\}$, and the projection operation is given by
\begin{equation}
	\mathrm{Proj}_{\tilde{\cv}_{(\!n\!)}}\!\big(\nabla_{\tilde{\cv}_{(n)}} g\big)  = \nabla _{\tilde{\cv}_{(n)}}g-\big(\tilde{\cv}_{(n)}^\TT\!\nabla_{\tilde{\cv}_{(n)}} g\big)\tilde{\cv}_{(n)}.
\end{equation}
We then update the optimization variable in the Riemannian gradient direction and finally retract it from the tangent space $\mathcal{T}_{\tilde{\cv}_{(n)}^{[j-1]}}\mathbb{M}$ onto the manifold~$\mathbb{M}$, thus obtaining a new candidate variable~$\tilde{\cv}_{(n)}^{[j]}$. These update and retraction operations in the unit 2-norm sphere can be implemented as
\begin{equation}
	\tilde{\cv}_{(n)}^{[j]} = \frac{\tilde{\cv}_{(n)}^{[j-1]}-\varepsilon \mathrm{Proj}_{\tilde{\cv}_{(n)}^{[j-1]}}\Big(\nabla_{\tilde{\cv}_{(n)}^{[j-1]}} g\Big)}{\Big\|\tilde{\cv}_{(n)}^{[j-1]}-\varepsilon \mathrm{Proj}_{\tilde{\cv}_{(n)}^{[j-1]}}\Big(\nabla_{\tilde{\cv}_{(n)}^{[j-1]}} g\Big)\Big\|_2},
\end{equation}
where $\varepsilon\in\mathbb{R}_+$ is the step size. This update procedure ensures that the iterations always proceed within the manifold~$\mathbb{M}$. 

\subsection{Algorithm Summary and Complexity Analysis}

\begin{algorithm}[t]
 \caption{Tri-Hybrid Precoding Based on Model~II}
 \label{algo:MII}
 \begin{algorithmic}[1]
 \State \textbf{Input:} $\Hm_k^\mathsf{cof},\sigma_k^2,\ \forall k$.\qquad \textbf{Output:} $\Fm_\mathsf{cof}^\mathsf{opt},\Fm_\mathsf{RF}^\mathsf{opt},\Fm_\mathsf{BB}^\mathsf{opt}$.
 \State Initialize $\Fm_\mathsf{cof},\Fm_\mathsf{RF},\Fm_\mathsf{BB}$, $\Fm_\Dt\!=\!\Fm_\mathsf{RF}\Fm_\mathsf{BB}$, $\Wm_k\!=\!\mathbf{I}$, {$f^{(0)}=\infty$,  convergence threshold $\epsilon$, and maximum iteration $I_\mathsf{max}$.}
 \For{$I=1,2,\dots,I_\mathsf{max}$}
 \State Compute $\Hm_k=\Hm_k^\mathsf{cof}\Fm_{\mathsf{cof}},\ \forall k$.
 \State Update $\Um_k=\Um_k^\mathsf{opt},\ \forall k$, according to~\eqref{eq:Uopt}.
 \State Update $\Wm_k=\Wm_k^\mathsf{opt},\ \forall k$, according to~\eqref{eq:Wopt}.
 \For{$n=1,2,\dots,N$}
 \State Compute $\{\Bm_{nn},\Qm_{n},\Dm_{n}\}$ via~\eqref{eq:Bm}--\eqref{eq:Dm} by replacing $\{\Hm_{(n),k}^\mathsf{sel},\bv_{(n)}\}$ with $\{\Hm_{(n),k}^\mathsf{cof},\cv_{(n)}\}$.
 \State Compute the optimal $\fv_{(n)}^\mathsf{opt}$ given $\cv_{(n)}$ via~\eqref{eq:fn_opt2}.
 \State Obtain $\cv_{(n)}^\mathsf{opt}$ given $\fv_{(n)}^\mathsf{opt}$ using the approach described in Section~\ref{sec:RieOpt}.
 \EndFor
 \State Recover $\Fm_\mathsf{cof}^\mathsf{opt}$ and $\Fm_\Dt^\mathsf{opt}$ based on \( \big\{ \fv_{(n)}^\mathsf{opt}, \cv_{(n)}^\mathsf{opt} \big\}_{n=1}^N \) using~\eqref{eq:Fcof_c} and~\eqref{eq:FDexpres}, respectively.
 \State {$f^{(I)}=\sum_{k=1}^K\beta_k \big(\mathrm{Tr}(\Wm_k\Em_k)-\ln\det(\Wm_k)\big)$.}
 \State {If $I>1$ and $\big|f^{(I)}-f^{(I-1)}\big|<\epsilon$, break.}
 \EndFor 
 \State Decompose $\Fm_\Dt^\mathsf{opt}$ back into $\Fm_{\mathsf{RF}}^\mathsf{opt}$ and $\Fm_{\mathsf{BB}}^\mathsf{opt}$ using the method described in Section~\ref{sec:decomp}.
 \end{algorithmic} 
\end{algorithm}

As a result, the proposed tri-hybrid precoding algorithm based on Model~II is summarized in Algorithm~\ref{algo:MII}. The computational complexity of Algorithm~\ref{algo:MII} is analyzed as follows. Under the same assumption that \( N \gg \sum_{k=1}^K M_k > D > K \), we focus on the complexity with respect to the number of \ac{BS} antennas \(N\) and the spherical harmonics truncation length \(T\). The complexity of lines~4--6 is dominated by line~4, which has a complexity of \( \mathcal{O}(N^2T) \). Over the loop in line~7, the operation in line~8 incurs a total complexity of \( \mathcal{O}(NT^2) \), the updates in lines~9 contribute a total complexity of \( \mathcal{O}(NT^2) \), and the Riemannian manifold optimization method in line~10 incurs a total complexity of \( \mathcal{O}(NT^2) \). Additionally, the decomposition method in line~17 has a complexity of \( \mathcal{O}(N_\mathsf{RF}^2N) \) when Algorithm~1 in~\cite{Huang2024Hybrid} is adopted. Therefore, the overall complexity of Algorithm~\ref{algo:MII} can be summarized as $\mathcal{O}\big(I_\mathsf{max}(N^2T + NT^2 + N_\mathsf{RF}^2N)\big)$.

\section{Simulation Results}\label{sec:sims}

This section presents a set of simulation results to show the performance of the considered tri-hybrid architecture and the proposed precoding algorithms.

\subsection{Simulation Setup}

To assess system performance in realistic scenarios, we generate the multi-user channels using ray-tracing simulation, as illustrated in Fig.~\ref{fig_RT}. The \ac{BS} is equipped with a $10 \times 10$ half-wavelength spaced antenna array. Three users are deployed in the scene, each configured with a $2 \times 2$ half-wavelength spaced antenna array. The ray-tracing simulations are carried out using Sionna RT in an environment based on the area around the Frauenkirche in Munich, Germany~\cite{sionna}. To accurately capture spatial characteristics, we perform per-antenna ray tracing, allowing the channel parameters for each antenna element to be independently simulated. The simulation includes both \ac{LoS} paths and \ac{NLoS} paths with up to triple-bounce reflections, as shown in Fig.~\ref{fig_RT}. From these simulations, we extract the complex channel gains, propagation delays, \acp{AoD}, and \acp{AoA} for each per-antenna link. These parameters are then used to construct the wireless channel according to~\eqref{eq:genChan}, and to compute the effective channels as defined in~\eqref{eq:Heff1} and~\eqref{eq:Heff2}. We first conduct simulations under perfect \ac{CSI}, and subsequently evaluate robustness to imperfect \ac{CSI} in Section~\ref{sec:ICSI}.

\begin{figure}[t]
  \centering
  \includegraphics[width=0.90\linewidth]{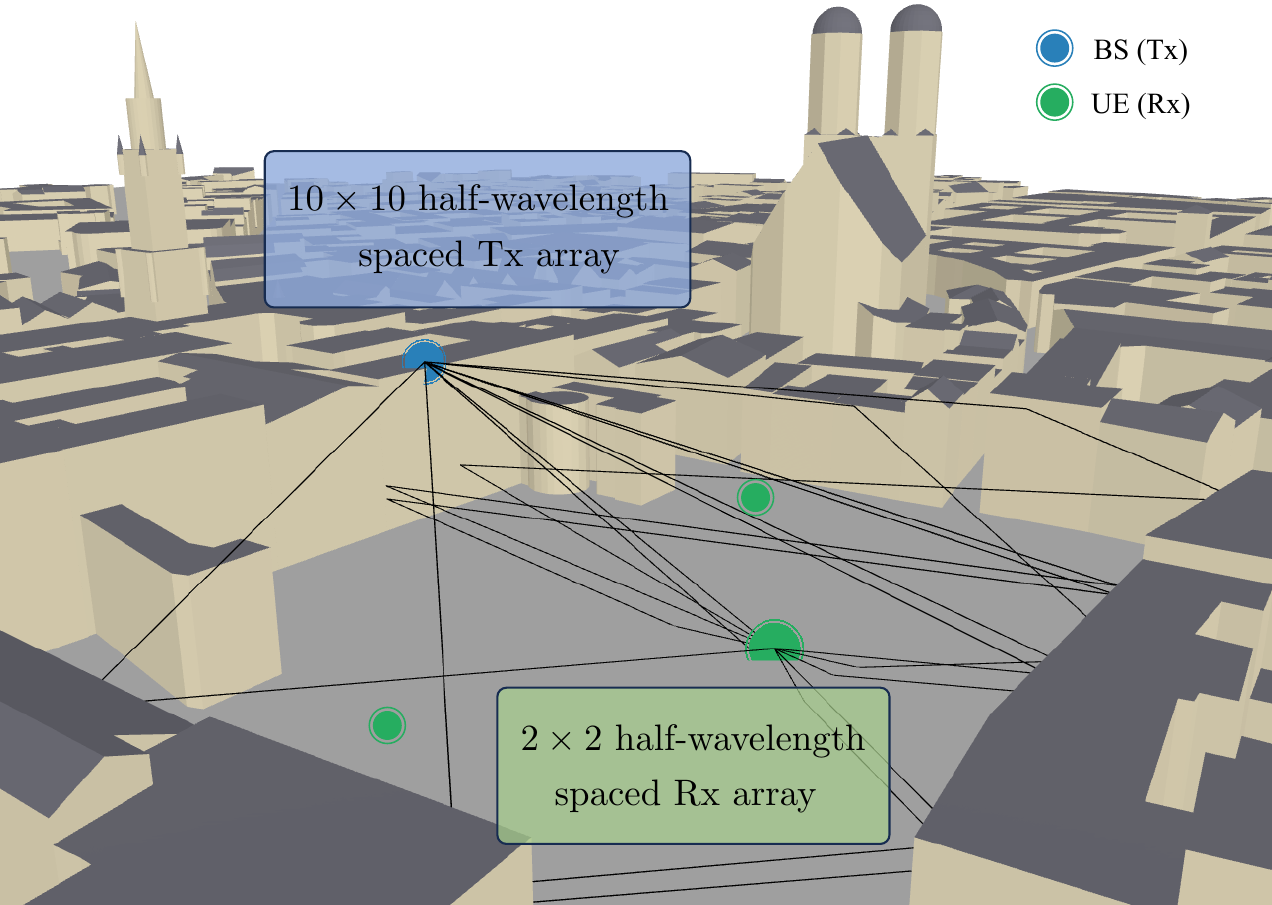}
 \caption{
Ray-tracing results generated using Sionna RT in an example scene set in Munich, Germany. The scenario consists of a BS with a $10 \times 10$ antenna array and three UEs, each equipped with a $2 \times 2$ antenna array. Both the BS and users are configured with half-wavelength spacing. For clarity, ray-tracing results are depicted for one representative UE only.}
  \label{fig_RT}
\end{figure}

\begin{figure}[t]
    \centering
%
%
\definecolor{mycolor1}{rgb}{0.13333,0.54510,0.13333}%
\definecolor{mycolor2}{rgb}{0.85490,0.64706,0.12549}%
\definecolor{mycolor3}{rgb}{1.00000,0.49804,0.05490}%
\definecolor{mycolor4}{rgb}{0.25490,0.41176,0.88235}%
\definecolor{mycolor5}{rgb}{0.60000,0.19608,0.80000}%
\begin{tikzpicture}

\begin{axis}[%
width=2.8in,
height=2.1in,
at={(0in,0in)},
scale only axis,
xmin=-2,
xmax=150,
xlabel style={font=\color{white!15!black},font=\footnotesize},
xticklabel style = {font=\color{white!15!black},font=\footnotesize},
xtick = {0,20,40,60,80,100,120,140},
xlabel={Outer Iteration $I$},
ymin=0,
ymax=13.5,
ytick = {0,2,4,6,8,10,12,14},
ylabel style={font=\color{white!15!black},font=\footnotesize, yshift=-13pt},
yticklabel style = {font=\color{white!15!black},font=\footnotesize},
ylabel={Weighted Sum-Rate (bps/Hz)},
axis background/.style={fill=white},
xmajorgrids,
ymajorgrids,
grid style={dashed},
legend style={
    at={(1,0)},
    anchor=south east,
    font=\footnotesize,
    legend cell align=left,
    align=left,
    draw=white!15!black,
    fill opacity=0.85
  }
]
\addplot [color=mycolor1, line width=1.1pt, mark=triangle*, mark options={solid, mycolor1}, mark size=1.8pt, mark repeat=13]
  table[row sep=crcr]{%
0	0.824498083796776\\
1	6.66859705247382\\
2	8.63082859439568\\
3	9.05713900420568\\
4	9.32136085632241\\
5	9.54577278911333\\
6	9.74258331671056\\
7	9.91759414904635\\
8	10.0747433205957\\
9	10.2174854441841\\
10	10.3482416418401\\
11	10.4688278408058\\
12	10.5807873404964\\
13	10.6853199765738\\
14	10.7832121538647\\
15	10.8749254521853\\
16	10.9607034248773\\
17	11.0407005761585\\
18	11.1151310950649\\
19	11.1843489079361\\
20	11.2488308203938\\
21	11.3091035621439\\
22	11.3656640850267\\
23	11.4189453754251\\
24	11.4693181563174\\
25	11.5171008137983\\
26	11.5625666876732\\
27	11.6059480002264\\
28	11.6474392892002\\
29	11.6871997815934\\
30	11.7253567692038\\
31	11.7620096301339\\
32	11.7972356580754\\
33	11.8310981243945\\
34	11.8636551207237\\
35	11.8949659169038\\
36	11.9250928777357\\
37	11.9541002659342\\
38	11.9820525965969\\
39	12.0090132810923\\
40	12.0350433739941\\
41	12.060200593936\\
42	12.0845384242635\\
43	12.1081057603582\\
44	12.1309469409697\\
45	12.1531019486999\\
46	12.1746068559477\\
47	12.1954941807661\\
48	12.2157933271595\\
49	12.2355309768492\\
50	12.2547314073981\\
51	12.2734168742241\\
52	12.2916079824275\\
53	12.3093239559142\\
54	12.3265829727651\\
55	12.3434024513264\\
56	12.359799276258\\
57	12.3757900572884\\
58	12.3913913365236\\
59	12.4066195431033\\
60	12.4214907635167\\
61	12.4360205522868\\
62	12.4502235343254\\
63	12.4641132008706\\
64	12.4777017760762\\
65	12.4910002231423\\
66	12.5040184061626\\
67	12.5167651490481\\
68	12.5292484868752\\
69	12.5414758639045\\
70	12.5534541384102\\
71	12.5651898815831\\
72	12.5766893487307\\
73	12.5879586208222\\
74	12.5990036842371\\
75	12.609830437672\\
76	12.620444719146\\
77	12.6308523463224\\
78	12.6410591096685\\
79	12.6510707812152\\
80	12.6608931015323\\
81	12.6705317632775\\
82	12.6799924014206\\
83	12.6892805641165\\
84	12.6984016672995\\
85	12.7073610400179\\
86	12.7161638461429\\
87	12.724815102516\\
88	12.7333196256414\\
89	12.7416820767619\\
90	12.7499068933828\\
91	12.7579982870159\\
92	12.7659602260566\\
93	12.7737964433898\\
94	12.7815103460275\\
95	12.789105073025\\
96	12.79658343382\\
97	12.8039478924692\\
98	12.8112005539774\\
99	12.8183431413385\\
100	12.825377017647\\
101	12.8323031294473\\
102	12.8391220353538\\
103	12.8458338987401\\
104	12.8524385743349\\
105	12.8589355666548\\
106	12.8653241577646\\
107	12.8716034097749\\
108	12.8777722582823\\
109	12.8838296956754\\
110	12.889774689566\\
111	12.8956063396739\\
112	12.9013239376249\\
113	12.9069270327862\\
114	12.9124155275577\\
115	12.9177897299606\\
116	12.9230503088049\\
117	12.9281983055673\\
118	12.9332351262471\\
119	12.938162541888\\
120	12.9429825974736\\
121	12.9476975661912\\
122	12.952309870029\\
123	12.9568220383633\\
124	12.9612366387658\\
125	12.9655562304431\\
126	12.9697833533258\\
127	12.9739204913931\\
128	12.9779700320752\\
129	12.9819342996774\\
130	12.9858155480276\\
131	12.9896159546178\\
132	12.9933376154959\\
133	12.996982592518\\
134	13.0005528759693\\
135	13.0040504224723\\
136	13.0074771405979\\
137	13.010834922438\\
138	13.0141256163228\\
139	13.0173510318601\\
140	13.0205129318176\\
141	13.0236130570057\\
142	13.0266531095713\\
143	13.0296347507519\\
144	13.0325596103371\\
145	13.035429264087\\
146	13.0382452521955\\
147	13.0410090602957\\
};
\addlegendentry{Algorithm~2 (Model II)}

\addplot [color=mycolor3, line width=1.1pt, mark=*, mark options={solid, mycolor3}, mark size=1.5pt, mark repeat=13]
  table[row sep=crcr]{%
0	0.824615414108352\\
1	5.39335763794878\\
2	6.45402014400237\\
3	6.96286556836491\\
4	7.40635575051104\\
5	7.69170711254965\\
6	7.8401022611513\\
7	7.93761734610532\\
8	8.01704457561805\\
9	8.08756519703196\\
10	8.15278054527463\\
11	8.20827075401249\\
12	8.25948536157809\\
13	8.30579601067708\\
14	8.34834519881462\\
15	8.38945277625531\\
16	8.4273720881906\\
17	8.46798106071445\\
18	8.50576739485189\\
19	8.53769005035536\\
20	8.56663817659381\\
21	8.59436575245885\\
22	8.62042912594875\\
23	8.64487629150414\\
24	8.6680638503742\\
25	8.69201583944\\
26	8.71344963412057\\
27	8.73363677078846\\
28	8.75423829479366\\
29	8.77359515507298\\
30	8.79178318059984\\
31	8.80880184726986\\
32	8.82535394351909\\
33	8.84109453891044\\
34	8.85688990963573\\
35	8.87134099023052\\
36	8.88523199525156\\
37	8.89860158598876\\
38	8.91146217953288\\
39	8.92383285185653\\
40	8.93573555774613\\
41	8.94718928454393\\
42	8.95821056579342\\
43	8.9688145258625\\
44	8.97901515575805\\
45	8.98882543776192\\
46	8.99970590603263\\
47	9.00919526017426\\
48	9.0181684815187\\
49	9.02673261309243\\
50	9.03495588240893\\
51	9.04279446584199\\
52	9.05217219185108\\
53	9.06024355219916\\
54	9.06790955799414\\
55	9.07471357469098\\
56	9.08118487078225\\
57	9.08732815413077\\
58	9.09317945339393\\
59	9.09876358110875\\
60	9.10409505500787\\
61	9.1091866733288\\
62	9.11405072077982\\
63	9.11869865519342\\
64	9.12314126564897\\
65	9.12801427435263\\
66	9.13279802440503\\
67	9.13741934619532\\
68	9.14193098924995\\
69	9.14586601971235\\
70	9.14955317962078\\
71	9.15343309982941\\
72	9.15694496854158\\
73	9.16025921514978\\
74	9.16339949903903\\
75	9.16639382988718\\
76	9.16926129343356\\
77	9.17201626204539\\
78	9.1746701930103\\
};
\addlegendentry{Algorithm~1 (Model I)}

\addplot [color=black, line width=1pt]
  table[row sep=crcr]{%
0	0.824615414108352\\
1	4.65991937997454\\
2	5.46851755875034\\
3	5.79480953967934\\
4	6.00369711400455\\
5	6.15588024889181\\
6	6.27632555249165\\
7	6.37658869529835\\
8	6.46318169551813\\
9	6.5399852906917\\
10	6.60941531058808\\
11	6.67308745371453\\
12	6.73211937498164\\
13	6.78729539496969\\
14	6.83917482925412\\
15	6.88809655536146\\
16	6.93350722927362\\
17	6.9766794943334\\
18	7.01778323192731\\
19	7.05694148338522\\
20	7.09425835558299\\
21	7.12984865875289\\
22	7.16378436771708\\
23	7.1961299997756\\
24	7.22694647373917\\
25	7.25629592733998\\
26	7.28424303549266\\
27	7.31085418696079\\
28	7.33619619135577\\
29	7.36033523116198\\
30	7.38333619514302\\
31	7.40526228912462\\
32	7.4261747809253\\
33	7.44613278289415\\
34	7.46519303373138\\
35	7.4834096820419\\
36	7.5008340938436\\
37	7.51751470929944\\
38	7.53349696618471\\
39	7.54882329524263\\
40	7.56353318089424\\
41	7.57766327292287\\
42	7.5912475316078\\
43	7.60431738955706\\
44	7.61690191674544\\
45	7.62902797960665\\
46	7.64072038941519\\
47	7.65200203894536\\
48	7.66289402911457\\
49	7.6734157888248\\
50	7.68358519152929\\
51	7.69341867138804\\
52	7.7029313405937\\
53	7.71213710797849\\
54	7.7210487977324\\
55	7.72967826622537\\
56	7.73803651461504\\
57	7.74613379506019\\
58	7.75397970877735\\
59	7.76158329467954\\
60	7.76895310777082\\
61	7.77609728677999\\
62	7.78302361072409\\
63	7.78973954427637\\
64	7.7962522720577\\
65	7.8025687223139\\
66	7.80869558086945\\
67	7.81463929669542\\
68	7.8204060808157\\
69	7.82600190053673\\
70	7.83143247108859\\
71	7.83670324671239\\
72	7.84181941303714\\
73	7.84678588228384\\
74	7.85160729243649\\
75	7.85628801105527\\
76	7.86083214391458\\
77	7.86524354817711\\
78	7.86952584942255\\
79	7.87368246158494\\
80	7.87771660873629\\
81	7.88163134768255\\
82	7.88542959047346\\
83	7.88911412612068\\
84	7.89268764102303\\
85	7.89615273777423\\
86	7.89951195216361\\
87	7.90276776827159\\
88	7.90592263161967\\
89	7.90897896037297\\
90	7.91193915462435\\
91	7.91480560381683\\
92	7.9175806923921\\
93	7.92026680378148\\
94	7.92286632288113\\
95	7.92538163717236\\
96	7.92781513665755\\
};
\addlegendentry{WMMSE~\cite{Zhao2023Rethinking} (Fixed antennas)}

\end{axis}
\end{tikzpicture}%
    \vspace{-3em}
    \caption{{Convergence performance of the proposed optimization algorithms.}}
    \label{fig_sim08}
\end{figure}

Unless otherwise specified, the system parameters are set as follows. The carrier frequency is set to $\unit[30]{GHz}$ in the mmWave band. Each user is allocated two data streams, i.e., $D_k = 2,\ \forall k$. The number of \ac{RF} chains at the \ac{BS} is set to $N_\mathsf{RF} = D + 3 = \big(\sum_{k=1}^K D_k\big) + 3=9$. The per-antenna power budget at the \ac{BS} is $P_n = \unit[0]{dBm},\ \forall n$, and the noise power at each user is $\sigma_k^2 = \unit[-90]{dBm},\ \forall k$. Equal weights are assigned to the three users, i.e., $\beta_1 = \beta_2 = \beta_3 = 1/3$. For Model~I, we use the 64 available radiation patterns from the real pattern-reconfigurable antenna prototype in~\cite{Wang2025Electromagnetically}. For Model~II, antenna patterns are arbitrarily optimized based on spherical harmonics decomposition, with the isotropic component fixed at $\rho = 0.7$ and a truncation length of $T = 49$. When evaluating the conventional hybrid \ac{MIMO} benchmark, the reconfigurability of the antenna radiation patterns is disabled, and all antennas are assigned a fixed radiation pattern as State~2 in Fig.~\ref{fig_trihybrid}. {In the algorithm implementation, the convergence threshold is set to $\epsilon = 2.8 \times 10^{-3}$, and the maximum number of iterations is $I_\mathsf{max} = 400$.}

\subsection{Performance Evaluation}

\subsubsection{Convergence Performance}

{
We first evaluate the convergence behavior of the proposed algorithms. As shown in Fig.~\ref{fig_sim08}, both algorithms based on Model~I and Model~II converge within a reasonable number of iterations, with Algorithm~1 exhibiting faster convergence. This is because Algorithm~1 solves each antenna update as a discrete search over a finite set of 64 candidate patterns, enabling global optimality in finite steps. In contrast, Algorithm~2 entails continuous nonconvex optimization on a Riemannian manifold, where only local optimality can be ensured for each subproblem, leading to slower convergence. However, Algorithm~2 achieves a higher weighted sum-rate compared to Algorithm~1, demonstrating the performance gain afforded by the additional degrees of freedom in Model~II. Moreover, compared with the fixed-antenna benchmark optimized via the \ac{WMMSE} algorithm~\cite{Zhao2023Rethinking}, both Algorithm~\ref{algo:MI} and Algorithm~\ref{algo:MII} demonstrate clear performance improvements while requiring a comparable number of iterations. 
}

\subsubsection{Antenna Radiation Pattern Visualization}

\begin{figure*}[t]
  \centering
  \includegraphics[width=0.8\linewidth]{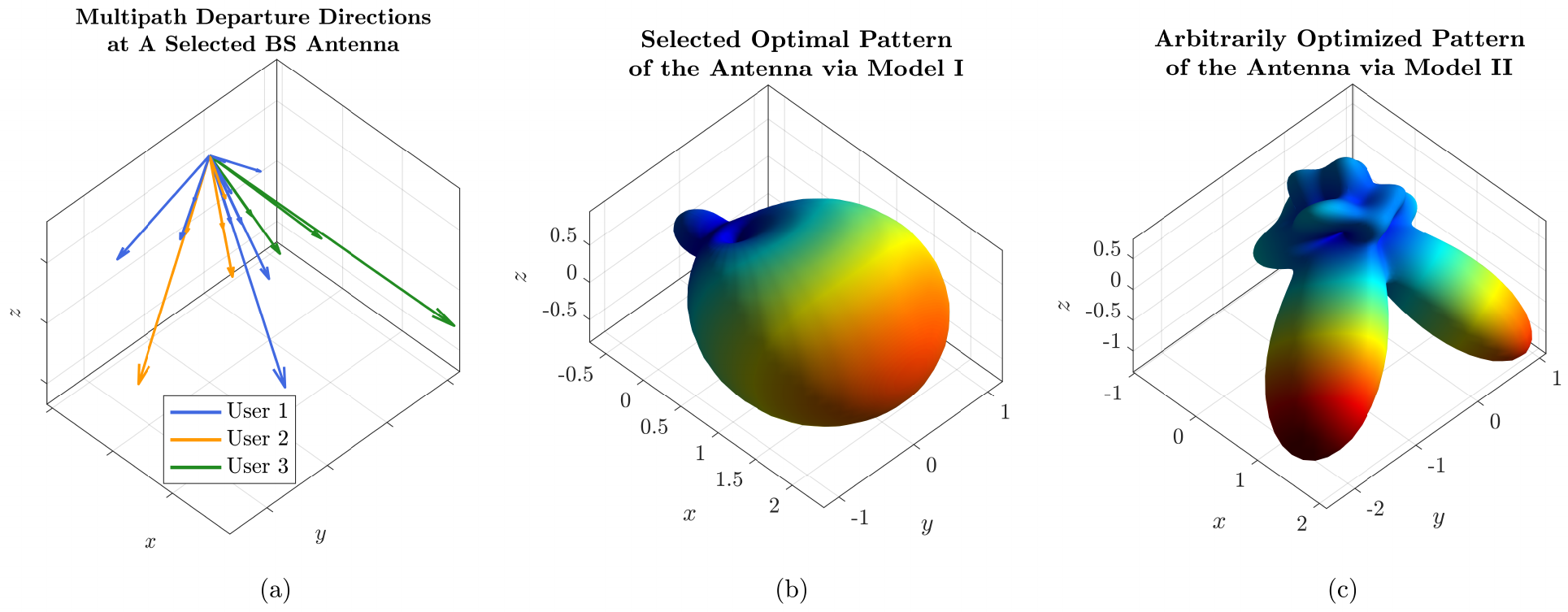}
  \vspace{-0.7em}
  \caption{
Visualization of the multipath geometry and the optimized antenna radiation patterns for a selected antenna element at the \ac{BS}, under the two models. (a) Multipath departure directions from the selected antenna to the three users, where the length of each arrow is inversely proportional to the path loss, i.e., proportional to the signal strength. (b) Optimized radiation pattern based on Model~I. (c) Optimized radiation pattern based on Model~II.
}
  \label{fig_sim03}
\end{figure*}

In Fig.~\ref{fig_sim03}-(a), we visualize the multipath departure directions from a selected \ac{BS} antenna to the three users. The length of each arrow is inversely proportional to the path loss, i.e., proportional to the signal strength. Figures~\ref{fig_sim03}-(b) and \ref{fig_sim03}-(c) plot the optimized radiation patterns at this antenna based on Model~I and Model~II, respectively. In Model~I, the proposed algorithm successfully selects the optimal candidate pattern that directs power toward the dominant multipath directions. However, this approach is limited by the constrained degrees of freedom inherent in antenna hardware designs. In contrast, the arbitrary optimization results under Model~II yield a more tailored radiation pattern that more precisely steers power toward the strongest multipath directions while also accounting for weaker paths. From an intuitive perspective, the proposed algorithms based on both models return reasonable and meaningful results.

\begin{figure}[t]
    \centering
    \include{figs/fig_sim03B.tex}
    \vspace{-3em}
    \caption{Visualization of the array beampattern under different antenna setups. The antenna gains are normalized with respect to their maximum value. {Here, $R_1$, $R_2$, and $R_3$ denote the achieved rate for the corresponding users.}}
    \label{fig_sim03B}
\end{figure}

\subsubsection{Array Beampattern Visualization}

\begin{figure}[t]
    \centering
%
%
\definecolor{mycolor1}{rgb}{0.13333,0.54510,0.13333}%
\definecolor{mycolor2}{rgb}{0.85490,0.64706,0.12549}%
\definecolor{mycolor3}{rgb}{1.00000,0.49804,0.05490}%
\definecolor{mycolor4}{rgb}{0.25490,0.41176,0.88235}%
\definecolor{mycolor5}{rgb}{0.60000,0.19608,0.80000}%
\begin{tikzpicture}

\begin{axis}[%
width=2.9in,
height=2.05in,
at={(0in,0in)},
scale only axis,
xmin=-20.5,
xmax=10.5,
xlabel style={font=\color{white!15!black},font=\footnotesize},
xticklabel style = {font=\color{white!15!black},font=\footnotesize},
xlabel={BS Per-Antenna Power Budget $P_n$ (dBm)},
ymin=-0.4,
ymax=20,
ylabel style={font=\color{white!15!black},font=\footnotesize, yshift=-13pt},
yticklabel style = {font=\color{white!15!black},font=\footnotesize},
ylabel={Weighted Sum-Rate (bps/Hz)},
axis background/.style={fill=white},
xmajorgrids,
ymajorgrids,
grid style={dashed},
legend columns=2,
legend style={
    at={(0.5,1.03)},
    anchor=south,
    font=\tiny,
    legend cell align=left,
    align=left,
    draw=white!15!black,
    fill opacity=0.85
  }
]
\addplot [color=mycolor1, dashed, line width=1pt, mark=triangle*, mark options={solid, mycolor1}, mark size=1.6pt]
  table[row sep=crcr]{%
-20	3.47426882816491\\
-15	5.02970325388622\\
-10	6.95353901140211\\
-5	9.51955347407958\\
0	13.0827188763107\\
5	15.9259377701614\\
10	18.7062761688731\\
};
\addlegendentry{$\mathbf{F}_\mathsf{D}^\mathsf{opt},\mathbf{F}_\mathsf{cof}^\mathsf{opt}$ (Model II)}

\addplot [color=mycolor1, line width=1pt, mark=triangle*, mark options={solid, mycolor1}, mark size=1.6pt]
  table[row sep=crcr]{%
-20	3.012579732559\\
-15	4.9663274854079\\
-10	6.94611042633265\\
-5	9.44760914189883\\
0	12.5286907728677\\
5	15.4149743366269\\
10	17.9090771716863\\
};
\addlegendentry{$\mathbf{F}_\mathsf{BB}^\mathsf{opt},\mathbf{F}_\mathsf{RF}^\mathsf{opt},\mathbf{F}_\mathsf{cof}^\mathsf{opt}$ (Model II)}

\addplot [color=mycolor2, dashed, line width=1pt, mark=diamond*, mark options={solid, rotate=90, mycolor2}, mark size=1.7pt]
  table[row sep=crcr]{%
-20	3.21183394692696\\
-15	4.77427092219899\\
-10	6.53710836740263\\
-5	8.47709297693073\\
0	10.2966571243027\\
5	12.5121332099337\\
10	14.6651395528531\\
};
\addlegendentry{$\mathbf{F}_\mathsf{D}^\mathsf{opt},\mathbf{F}_\mathsf{sel}^\mathsf{opt}$ (Model I, fictitious)}

\addplot [color=mycolor2, line width=1pt, mark=diamond*, mark options={solid, rotate=90, mycolor2}, mark size=1.7pt]
  table[row sep=crcr]{%
-20	2.66012993748375\\
-15	4.40511061161033\\
-10	6.51430115858512\\
-5	8.45107117433742\\
0	10.1288998831463\\
5	12.3562770135309\\
10	14.2098637006697\\
};
\addlegendentry{$\mathbf{F}_\mathsf{BB}^\mathsf{opt},\mathbf{F}_\mathsf{RF}^\mathsf{opt},\mathbf{F}_\mathsf{sel}^\mathsf{opt}$ (Model I, fictitious)}

\addplot [color=mycolor3, dashed, line width=1pt, mark=*, mark options={solid, mycolor3}, mark size=1.3pt]
  table[row sep=crcr]{%
-20	2.10403961707919\\
-15	3.45683763166029\\
-10	5.19159301630005\\
-5	7.21250347307435\\
0	8.68464431460263\\
5	10.8297650765822\\
10	13.0469378084869\\
};
\addlegendentry{$\mathbf{F}_\mathsf{D}^\mathsf{opt},\mathbf{F}_\mathsf{sel}^\mathsf{opt}$ (Model I, hardware)}

\addplot [color=mycolor3, line width=1pt, mark=*, mark options={solid, mycolor3}, mark size=1.3pt]
  table[row sep=crcr]{%
-20	2.10292628292906\\
-15	3.43503376972166\\
-10	5.16857352109015\\
-5	7.18996476693722\\
0	8.59920502511259\\
5	10.6900215983374\\
10	12.7866968182182\\
};
\addlegendentry{$\mathbf{F}_\mathsf{BB}^\mathsf{opt},\mathbf{F}_\mathsf{RF}^\mathsf{opt},\mathbf{F}_\mathsf{sel}^\mathsf{opt}$ (Model I, hardware)}

\addplot [color=mycolor4, dashed, line width=1pt, mark=+, mark options={solid, mycolor4}, mark size=2.3pt]
  table[row sep=crcr]{%
-20	1.90914374814614\\
-15	3.17385386521425\\
-10	4.84013954885365\\
-5	6.80887913592263\\
0	7.92781513665756\\
5	10.0677119018245\\
10	12.4252474693253\\
};
\addlegendentry{$\mathbf{F}_\mathsf{D}^\mathsf{opt}$ (WMMSE) + fixed antennas}

\addplot [color=mycolor4, line width=1pt, mark=+, mark options={solid, mycolor4}, mark size=2.3pt]
  table[row sep=crcr]{%
-20	1.90654914599952\\
-15	3.13484068889517\\
-10	4.8147970356933\\
-5	6.77599980141731\\
0	7.81535828047099\\
5	9.8145839499693\\
10	12.1098754169075\\
};
\addlegendentry{$\mathbf{F}_\mathsf{BB}^\mathsf{opt},\mathbf{F}_\mathsf{RF}^\mathsf{opt}$ (WMMSE) + fixed antennas}

\addplot [color=mycolor5, dashed, line width=1pt, mark=x, mark options={solid, mycolor5}, mark size=2.3pt]
  table[row sep=crcr]{%
-20	0.330808888172056\\
-15	0.884265163070366\\
-10	1.70445212054909\\
-5	2.72232865890255\\
0	4.05016055218815\\
5	6.36552190631533\\
10	8.19550332071532\\
};
\addlegendentry{$\mathbf{F}_\mathsf{D}^\mathsf{opt}$ (ZF) + fixed antennas}

\addplot [color=mycolor5, line width=1pt, mark=x, mark options={solid, mycolor5}, mark size=2.3pt]
  table[row sep=crcr]{%
-20	0.0656453515109489\\
-15	0.2494002520509\\
-10	0.66467876668446\\
-5	1.37711412889063\\
0	2.26094552915706\\
5	3.89171500539129\\
10	5.75857767451646\\
};
\addlegendentry{$\mathbf{F}_\mathsf{BB}^\mathsf{opt},\mathbf{F}_\mathsf{RF}^\mathsf{opt}$ (ZF) + fixed antennas}

\end{axis}
\end{tikzpicture}%
    \vspace{-3em}
    \caption{Evaluation of weighted sum-rate versus \ac{BS} per-antenna power budget.}
    \label{fig_sim04}
\end{figure}

Subsequently, we evaluate the array beampatterns optimized under different antenna configurations. Following the formulation in~\cite[Eq.~(9)--(11)]{Zheng2024Mutual}, the array beampattern for the $k^\text{th}$ user is computed as
$
	E_k(\theta,\phi) = \|\rv^\TT(\theta,\phi)\Fm_\mathsf{RF}\Fm_{\mathsf{BB},k}\|_2, 
$
where $\rv(\theta,\phi) \in \mathbb{C}^N$ denotes the array manifold vector. Specifically, the $n^\text{th}$ entry of $\rv(\theta,\phi)$ is given by $r_n(\theta,\phi)=G_{(n)}^\mathsf{BS}(\theta,\phi)e^{j\frac{2\pi}{\lambda}\pv_n^\TT\uv(\theta,\phi)}$, where $\pv_n$ denotes the position of the $n^\text{th}$ antenna in the BS's body coordinate system, and $\uv(\theta,\phi)$ is a unit directional vector defined as $\uv(\theta,\phi)=[\sin\theta\cos\phi,\sin\theta\sin\phi,\cos\theta]^\TT$. In the three subfigures of Fig.~\ref{fig_sim03B}, the BS antenna radiation patterns $G_{(n)}^\mathsf{BS}(\theta,\phi)$ (for all $n$) are configured as fixed, optimized using Model~I, and optimized using Model~II, respectively. To offer an intuitive \ac{2D} illustration, Fig.~\ref{fig_sim03B} plots the envelope of the beampattern over the azimuth angle $\phi$, which is obtained as $\tilde{E}_k(\phi) = \max_\theta E_k(\theta,\phi)$. 

While the multi-user precoding objective cannot be directly visualized due to inter-user interference, some intuitive insights can still be gleaned from the figure. In general, the additional degrees of freedom provided by reconfigurable antenna radiation patterns allow the array beampattern to allocate power more effectively toward the dominant multipath directions. For instance, in the fixed-antenna case, the system fails to leverage the \ac{LoS} paths to User~3, resulting in its normalized gains being lower than $\unit[-10]{dB}$. In contrast, under the reconfigurable cases, particularly Model~II, the \ac{LoS} path to User~3 receives significantly more power, leading to a substantially higher system sum-rate. {In addition, the limitations of the proposed precoding algorithm under the \ac{WMMSE} framework are also evident here. For example, in all three cases, the beampatterns for User~1 and User~3 fail to direct their main lobes toward reasonable directions.  This suboptimality arises from the nonconvex nature of the problem, for which the \ac{BCD} methods cannot guarantee global optimality. Thus, there still remains potential for further improvement from the algorithmic perspective.}

\subsubsection{Weighted Sum-Rate versus Transmit Power}
Next, we evaluate the performance gain offered by the proposed tri-hybrid \ac{MIMO} architecture in comparison to the conventional hybrid \ac{MIMO} system. Figure~\ref{fig_sim04} demonstrates the weighted sum-rate as a function of the BS per-antenna power budget across different system architectures. Alongside the tri-hybrid architecture optimized under both Model~I and Model~II, we also include the conventional hybrid architecture with fixed antennas, optimized using two benchmark methods: the \ac{WMMSE} precoding algorithm~\cite{Zhao2023Rethinking} and the \ac{ZF} algorithm~\cite{Spencer2004Zero}, all under the same per-antenna power constraint. The WMMSE algorithm jointly optimizes the analog and digital precoders and is widely regarded as a near-optimal approach, while the ZF algorithm provides precoders that cancel mutual interference. Note that under Model~I, besides the set of 64 real antenna patterns in~\cite{Wang2025Electromagnetically} (denoted as ``Model~I, hardware''), we further add a benchmark set of 64 fictitious antenna patterns (denoted as ``Model~I, fictitious''). These fictitious patterns are set as Gaussian beams steered toward 64 directions uniformly distributed within $\theta \in [90^\circ, 180^\circ]$ and $\phi \in [-90^\circ, 90^\circ]$, each with a 3-dB beamwidth of $85^\circ$. 

We present both the fully digital solutions and those obtained by decomposing the digital precoder back into digital and hybrid components. The results indicate that the decomposition method proposed in Section~\ref{sec:decomp} effectively preserves the overall weighted sum-rate, with only minor performance degradation for the tri-hybrid solutions and the WMMSE-based hybrid solution. Regarding the comparison between system architectures, the tri-hybrid architecture consistently outperforms the conventional hybrid design; however, the Model~I-based solution based on real hardware provides only a minor gain of up to $\unit[0.9]{bps/Hz}$, while the Model~II-based solution delivers a much more substantial improvement of up to $\unit[5.8]{bps/Hz}$. {This highlights the limited performance gains achievable with the existing pattern-reconfigurable antenna hardware in~\cite{Wang2025Electromagnetically} due to its restricted flexibility. The pronounced gap between Model~I and Model~II (i) reveals substantial untapped potential of tri-hybrid MIMO and (ii) underscores the need for more flexible antenna designs.} Moreover, the performance of Model~I with fictitious patterns suggests that arbitrary pattern generation may not be strictly necessary. Substantial gains can also be achieved by designing a more effective set of limited patterns. 

\begin{figure}[t]
    \centering
%
%
\definecolor{mycolor1}{rgb}{0.13333,0.54510,0.13333}%
\definecolor{mycolor2}{rgb}{0.85490,0.64706,0.12549}%
\definecolor{mycolor3}{rgb}{1.00000,0.49804,0.05490}%
\definecolor{mycolor4}{rgb}{0.25490,0.41176,0.88235}%
\definecolor{mycolor5}{rgb}{0.60000,0.19608,0.80000}%
\begin{tikzpicture}

\begin{axis}[%
width=2.8in,
height=2.1in,
at={(0in,0in)},
scale only axis,
xmin=-0.09,
xmax=5.09,
xlabel style={font=\color{white!15!black},font=\footnotesize},
xticklabel style = {font=\color{white!15!black},font=\footnotesize},
xtick = {0,1,2,3,4,5},
xticklabels = {$D$, $D+1$, $D+2$, $D+3$, $D+4$, $D+5$},
xlabel={Number of BS \ac{RF} Chains $N_\mathsf{RF}$},
ymin=9.5,
ymax=19,
ylabel style={font=\color{white!15!black},font=\footnotesize, yshift=-13pt},
yticklabel style = {font=\color{white!15!black},font=\footnotesize},
ylabel={Weighted Sum-Rate (bps/Hz)},
axis background/.style={fill=white},
xmajorgrids,
ymajorgrids,
grid style={dashed},
legend style={
    at={(0,1)},
    anchor=north west,
    font=\tiny,
    legend cell align=left,
    align=left,
    draw=white!15!black,
    fill opacity=0.85
  }
]
\addplot [color=mycolor1, line width=1pt, mark=triangle*, mark options={solid, mycolor1}, mark size=1.9pt]
  table[row sep=crcr]{%
0	11.8483583616531\\
1	13.6926566513961\\
2	15.4666420903395\\
3	17.9090771716863\\
4	18.4463729168961\\
5	18.6473454768515\\
};
\addlegendentry{Tri-hybrid (Model II, $\rho=0.7$)}

\addplot [color=mycolor1, line width=1pt, mark=square*, mark options={solid, mycolor1}, mark size=1.3pt]
  table[row sep=crcr]{%
0	11.6056726668707\\
1	12.9649967391942\\
2	14.4819877352779\\
3	16.465094477283\\
4	16.839699173221\\
5	16.9171336333246\\
};
\addlegendentry{Tri-hybrid (Model II, $\rho=0.9$)}

\addplot [color=mycolor3, line width=1pt, mark=*, mark options={solid, mycolor3}, mark size=1.6pt]
  table[row sep=crcr]{%
0	10.3705761657611\\
1	11.5924976638508\\
2	12.4393053325297\\
3	12.7866968182182\\
4	12.9221321591952\\
5	12.9994472399703\\
};
\addlegendentry{Tri-hybrid (Model I, hardware)}

\addplot [color=mycolor4, line width=1pt, mark=+, mark options={solid, mycolor4}, mark size=2.5pt]
  table[row sep=crcr]{%
0	9.66192050787265\\
1	10.9044421178127\\
2	11.7658651694964\\
3	12.1098754169075\\
4	12.3117136620451\\
5	12.3971724333544\\
};
\addlegendentry{Hybrid (WMMSE)}

\draw[<->, thick, black] (axis cs:2.1,12.4) -- (axis cs:4.9,12.4);
\draw[-, thick, black] (axis cs:3.5,12.4) -- (axis cs:3.3,11.2);
\node[below, black] at (axis cs:3.3,11.2) {\scriptsize{Reduce 3 RF chains}};

\draw[<->, thick, black] (axis cs:1.1,12.7) -- (axis cs:4.9,12.7);
\draw[-, thick, black] (axis cs:2.6,12.7) -- (axis cs:2.8,13.9);
\node[above right, black] at (axis cs:2.3,13.9) {\scriptsize{Reduce 4 RF chains}};

\end{axis}
\end{tikzpicture}%
    \vspace{-3em}
    \caption{Evaluation of weighted sum-rate versus number of \ac{BS} \ac{RF} chains, where $D=6$ is the total number of data streams required by all users.}
    \label{fig_sim05}
\end{figure}

\begin{figure}[t]
    \centering
%
%
\definecolor{mycolor1}{rgb}{0.13333,0.54510,0.13333}%
\definecolor{mycolor2}{rgb}{0.85490,0.64706,0.12549}%
\definecolor{mycolor3}{rgb}{1.00000,0.49804,0.05490}%
\definecolor{mycolor4}{rgb}{0.25490,0.41176,0.88235}%
\definecolor{mycolor5}{rgb}{0.60000,0.19608,0.80000}%
\begin{tikzpicture}

\begin{axis}[%
width=2.8in,
height=2.1in,
at={(0in,0in)},
scale only axis,
xmin=8,
xmax=102,
xlabel style={font=\color{white!15!black},font=\footnotesize},
xticklabel style = {font=\color{white!15!black},font=\footnotesize},
xtick = {10,20,30,40,50,60,70,80,90,100},
xlabel={Number of BS Antennas $N$},
ymin=3,
ymax=18.5,
ylabel style={font=\color{white!15!black},font=\footnotesize, yshift=-13pt},
yticklabel style = {font=\color{white!15!black},font=\footnotesize},
ylabel={Weighted Sum-Rate (bps/Hz)},
axis background/.style={fill=white},
xmajorgrids,
ymajorgrids,
grid style={dashed},
legend style={
    at={(1,0)},
    anchor=south east,
    font=\tiny,
    legend cell align=left,
    align=left,
    draw=white!15!black,
    fill opacity=0.85
  }
]
\addplot [color=mycolor1, line width=1pt, mark=triangle*, mark options={solid, mycolor1}, mark size=1.9pt]
  table[row sep=crcr]{%
10	6.19457072788849\\
25	9.83457344153823\\
40	13.7852098109186\\
55	15.5701168433036\\
70	16.7512423710616\\
85	17.3400078317156\\
100	17.9090771716863\\
};
\addlegendentry{Tri-hybrid (Model II, $\rho=0.7$)}

\addplot [color=mycolor1, line width=1pt, mark=square*, mark options={solid, mycolor1}, mark size=1.3pt]
  table[row sep=crcr]{%
10	5.26178996510742\\
25	8.48415259856495\\
40	11.7596527854929\\
55	13.6932013471459\\
70	14.8547351848762\\
85	15.5107255528956\\
100	16.465094477283\\
};
\addlegendentry{Tri-hybrid (Model II, $\rho=0.9$)}

\addplot [color=mycolor3, line width=1pt, mark=*, mark options={solid, mycolor3}, mark size=1.6pt]
  table[row sep=crcr]{%
10	3.80111508501751\\
25	6.65523338215881\\
40	9.69085992117057\\
55	10.9833372720943\\
70	11.4432864209644\\
85	11.8138079998274\\
100	12.7866968182182\\
};
\addlegendentry{Tri-hybrid (Model I, hardware)}

\addplot [color=mycolor4, line width=1pt, mark=+, mark options={solid, mycolor4}, mark size=2.5pt]
  table[row sep=crcr]{%
10	3.61389304013674\\
25	6.34765104798253\\
40	9.36403771522432\\
55	10.6772970813369\\
70	10.7673267885922\\
85	11.3454583891942\\
100	12.1098754169075\\
};
\addlegendentry{Hybrid (WMMSE)}

\draw[<->, thick, black] (axis cs:46,12.3) -- (axis cs:98,12.3);
\node[above right, black] at (axis cs:57,12.3) {\scriptsize{Reduce \unit[57]{\%} antennas}};
\draw[<->, thick, black] (axis cs:91,12) -- (axis cs:98,12);
\draw[-, thick, black] (axis cs:95,12) -- (axis cs:91,10);
\node[below, black] at (axis cs:84,10) {\scriptsize{Reduce \unit[10]{\%}} antennas};

\end{axis}

\end{tikzpicture}%
    \vspace{-3em}
    \caption{Evaluation of weighted sum-rate versus number of \ac{BS} antennas.}
    \label{fig_sim06}
\end{figure}

\subsubsection{Hardware Efficiency Evaluation}
Apart from the performance gain, the tri-hybrid architecture holds strong potential for reducing reliance on costly system components. To validate this, Fig.~\ref{fig_sim05} and Fig.~\ref{fig_sim06} illustrate the weighted sum-rate as a function of the number of \ac{RF} chains $N_\mathsf{RF}$ and the number of antennas $N$ at the \ac{BS}, respectively. The results in both figures show that system performance generally improves with an increasing number of RF chains or antennas. In these two figures, we also evaluate the performance of Model~II under various values of isotropic harmonic strength $\rho$, as defined in~\eqref{eq:ImpliConst}. A lower value of $\rho$ corresponds to greater flexibility in spherical harmonics-based radiation pattern synthesis. It is evident that increased flexibility in radiation pattern synthesis leads to more significant enhancements in system sum-rate. Moreover, the Model~I-based tri-hybrid solution adopting the real hardware patterns can achieve the same level of system sum-rate while reducing three RF chains or approximately $\unit[10]{\%}$ of the antennas. In contrast, the Model~II-based tri-hybrid solution with $\rho = 0.9$ can reduce four RF chains or approximately $\unit[57]{\%}$ of the antennas without compromising performance. Besides, we also observe that when $N_\mathsf{RF}$ and $N$ are small, the hardware cost reduction provided by the tri-hybrid architecture is less pronounced. However, in these regions, since the hardware cost is already low, performance improvement becomes the primary concern. As shown in both figures, the proposed architecture still offers substantial performance gains in regions with low $N_\mathsf{RF}$ and $N$. These findings demonstrate the potential of the pattern-reconfigurable antenna-based tri-hybrid architecture in enabling low-cost and energy-efficient \ac{MIMO} communications. 

\subsubsection{Impact of Antenna Flexibility}

\begin{figure}[t]
    \centering
%
%
\definecolor{mycolor1}{rgb}{0.13333,0.54510,0.13333}%
\definecolor{mycolor2}{rgb}{0.85490,0.64706,0.12549}%
\definecolor{mycolor3}{rgb}{1.00000,0.49804,0.05490}%
\definecolor{mycolor4}{rgb}{0.25490,0.41176,0.88235}%
\definecolor{mycolor5}{rgb}{0.60000,0.19608,0.80000}%
\begin{tikzpicture}

\begin{axis}[%
width=2.7in,
height=2.1in,
at={(0in,0in)},
scale only axis,
xmin=4,
xmax=81,
xlabel style={font=\color{white!15!black},font=\footnotesize},
xticklabel style = {font=\color{white!15!black},font=\footnotesize},
xtick = {5,10,20,30,40,50,60,70,80},
xlabel={Number of Available Patterns $S$ per Antenna},
ymin=3.5,
ymax=6,
ylabel style={font=\color{white!15!black},font=\footnotesize, yshift=-12pt},
yticklabel style = {font=\color{white!15!black},font=\footnotesize},
ylabel={Weighted Sum-Rate (bps/Hz)},
axis background/.style={fill=white},
xmajorgrids,
ymajorgrids,
grid style={dashed},
legend style={
    at={(1,0)},
    anchor=south east,
    font=\tiny,
    legend cell align=left,
    align=left,
    draw=white!15!black,
    fill opacity=0.85
  }
]
\addplot [color=red, line width=1.2pt, dashed]
  table[row sep=crcr]{%
4	5.876471806494108\\
82 5.876471806494108\\
};
\addlegendentry{Model II reference ($\rho = 0.7$)}

\addplot [color=mycolor3, line width=1pt, mark=*, mark options={solid, mycolor3}, mark size=1.8pt]
  table[row sep=crcr]{%
5	3.63362498345802\\
20	5.28755061637645\\
35	5.32983855231356\\
50	5.36927286296817\\
65	5.40564832998449\\
80	5.41849330075639\\
};
\addlegendentry{Model I fictitious, 3-dB BW = $45^\circ$}

\addplot [color=mycolor1, line width=1pt, mark=square*, mark options={solid, mycolor1}, mark size=1.5pt]
  table[row sep=crcr]{%
5	4.28846872380926\\
20	4.54435149053953\\
35	4.67360622963269\\
50	4.68699019525517\\
65	4.71410339510417\\
80	4.73362481410829\\
};
\addlegendentry{Model I fictitious, 3-dB BW = $90^\circ$}

\addplot [color=black, line width=1.1pt]
  table[row sep=crcr]{%
4	4.422688090406239\\
82 4.422688090406239\\
};
\addlegendentry{Fixed-pattern hybrid reference}

\end{axis}

\end{tikzpicture}%
    \vspace{-3em}
    \caption{{Weighted sum-rate versus the number of available patterns $S$ per antenna in the ``Model I, fictitious’’ case.}}
    \label{fig_sim07}
\end{figure}

{
To further quantify the impact of antenna flexibility, we evaluate the weighted sum-rate as a function of the number of available patterns $S$ per antenna under the ``Model I, fictitious’’ setting. This evaluation is averaged over 100 random channel realizations. As shown in Fig.~\ref{fig_sim07}, increasing $S$ generally improves performance. Interestingly, when $S$ is small (e.g., $S<10$), pattern-reconfigurable antennas with directional Gaussian patterns underperform the fixed-pattern baseline using State~2, as the limited pattern set cannot adequately cover the angular domain for multipath diversity. As $S$ increases (approximately $S>12$), the performance surpasses that of fixed-pattern antennas, reflecting the benefit of enhanced spatial selectivity.

Moreover, the impact of beamwidth per antenna pattern exhibits a clear trade-off. For large $S$, narrower beams (higher directivity) yield higher sum-rate by enabling more precise power focusing toward dominant propagation directions, thereby enhancing effective channel gains. However, when $S$ is small, narrower beams also exacerbate the performance degradation due to insufficient angular coverage, whereas wider beams provide more robust performance by better accommodating multipath uncertainty. A final observation is that the performance of Model~I consistently remains below that of Model~II. This is because Model~II enables arbitrary pattern generation with the highest degrees of freedom, unconstrained by predefined directions, beamwidths, or pattern shapes.}

{
\subsubsection{Impact of Isotropic Constraint $\rho$ in Model~II} 

\begin{figure}[t]
    \centering
%
%
\definecolor{mycolor1}{rgb}{0.13333,0.54510,0.13333}%
\definecolor{mycolor2}{rgb}{0.85490,0.64706,0.12549}%
\definecolor{mycolor3}{rgb}{1.00000,0.49804,0.05490}%
\definecolor{mycolor4}{rgb}{0.25490,0.41176,0.88235}%
\definecolor{mycolor5}{rgb}{0.60000,0.19608,0.80000}%

\begin{tikzpicture}

\begin{axis}[
    width=3.3in,
    height=2.7in,
    at={(0in,0in)},
    xlabel={Isotropic Harmonic Constraint $\rho$},
    xtick=data,
    xlabel style={font=\color{white!15!black},font=\footnotesize},
    xticklabel style = {font=\color{white!15!black},font=\footnotesize},
    axis lines=box,
	ytick pos=left,
    ymin=9.5,
    ymax = 14.5,
    xmin = 0.49,
    xmax = 0.96,
    ylabel style={color=mycolor3,font=\footnotesize},
    yticklabel style = {color=mycolor3,font=\footnotesize},
    tick align=inside,
    ylabel style={color=black, at={(axis description cs:0.1,0.5)}, anchor=south},
    ylabel={\color{mycolor3} Weighted Sum-Rate (bps/Hz)},
    xmajorgrids,
    grid style={dashed},
    legend style={at={(0.03,0.97)}, anchor=north west}
]

\addplot+[color=mycolor3, line width=1pt, mark=*, mark options={solid, mycolor3}, mark size=1.8pt] coordinates {
(0.5,14.0079845214751) 
(0.6,13.5968407013607) 
(0.7,13.0827188451399) 
(0.8,12.3268206082548) 
(0.9,11.0347811293203) 
(0.95,9.87028500213263)};

\end{axis}

\begin{axis}[
    width=3.3in,
    height=2.7in,
    at={(0in,0in)},
    axis y line*=right,
    axis x line=none,
    ymin=-0.2, 
    ymax=6.95,
    xmin = 0.49,
    xmax = 0.96,
    ytick = {0,1,2,3,4,5,6,7},
    yticklabel style={color=red},
    tick align=inside,
    ylabel style={color=mycolor1, at={(axis description cs:1.33,0.5)}, anchor=south, font=\footnotesize},
    yticklabel style = {color=mycolor1,font=\footnotesize},
    ylabel={\color{mycolor1} Negative-Energy Ratio $\big(\times 10^{-3}\big)$},
	ymajorgrids,
	grid style={dashed},
    legend style={at={(1,1)}, anchor=north east, font=\scriptsize, legend cell align=left, align=left, draw=white!15!black, fill opacity=0.85}
]

\addplot+[color=mycolor3, line width=1pt, mark=*, mark options={solid, mycolor3}, mark size=1.8pt] coordinates {
(1,1) 
};
\addlegendentry{Optimized Weighted Sum-Rate};

\addplot+[color=mycolor1, line width=1pt, mark=square*, mark options={solid, mycolor1}, mark size=1.5pt] coordinates {
(0.5,6.69676109754422) 
(0.6,3.14662438072991) 
(0.7,1.03083615882705) 
(0.8,0.187279607941624) 
(0.9,0.00138598698759837) 
(0.95,0)};
\addlegendentry{Negative-Energy Ratio};

\end{axis}
\end{tikzpicture}
    \vspace{-2em}
    \caption{{Sensitivity of Model~II to the isotropic-component strength $\rho$. The left axis shows the optimized fully digital weighted sum-rate, and the right axis shows the negative-energy ratio of the synthesized BS antenna patterns.}}
    \label{fig_sim09}
\end{figure}

When using Model~II, a key factor is the selection of the isotropic constraint parameter $\rho$. As discussed in Section~\ref{sec:rho}, a relatively large $\rho$ is generally required to ensure the nonnegativity of the synthesized radiation pattern over the sphere, whereas a smaller $\rho$ means a more flexible antenna design and can lead to improved communication performance. To quantitatively evaluate this trade-off, Fig.~\ref{fig_sim09} presents the system sum-rate and the ratio of “negative energy” in all the synthesized radiation patterns across different values of $\rho$. Consistent with the aforementioned analysis, a smaller $\rho$ yields a higher weighted sum-rate after optimization using the proposed algorithm, but also increases the fraction of physically infeasible “negative radiation.” In practice, selecting $\rho \geq 0.7$ typically ensures that the negative-energy ratio remains below $0.1\%$.
}

{
\subsubsection{Impact of Imperfect CSI}\label{sec:ICSI}

\begin{figure}[t]
    \centering
%
%
\definecolor{mycolor1}{rgb}{0.13333,0.54510,0.13333}%
\definecolor{mycolor2}{rgb}{0.85490,0.64706,0.12549}%
\definecolor{mycolor3}{rgb}{1.00000,0.49804,0.05490}%
\definecolor{mycolor4}{rgb}{0.25490,0.41176,0.88235}%
\definecolor{mycolor5}{rgb}{0.60000,0.19608,0.80000}%
\begin{tikzpicture}

\begin{axis}[%
width=2.9in,
height=2.5in,
at={(0in,0in)},
scale only axis,
xmin=-0.3,
xmax=20.3,
xlabel style={font=\color{white!15!black},font=\footnotesize},
xticklabel style = {font=\color{white!15!black},font=\footnotesize},
xlabel={CSI Error Level $\eta$ (\%)},
ymin=6,
ymax=17.5,
ylabel style={font=\color{white!15!black},font=\footnotesize, yshift=-13pt},
yticklabel style = {font=\color{white!15!black},font=\footnotesize},
ylabel={Weighted Sum-Rate (bps/Hz)},
axis background/.style={fill=white},
xmajorgrids,
ymajorgrids,
grid style={dashed},
legend columns=2,
legend style={
    at={(1,1)},
    anchor=north east,
    font=\tiny,
    legend cell align=left,
    align=left,
    draw=white!15!black,
    fill opacity=0.85
}
]
\addplot [color=mycolor1, dashed, line width=1pt, mark=triangle*, mark options={solid, mycolor1}, mark size=1.6pt]
  table[row sep=crcr]{%
0	13.0827188451399\\
5	13.068511426315\\
10	13.0350052221443\\
15	12.9196470654833\\
20	12.7542230748079\\
};
\addlegendentry{$\mathbf{F}_\mathsf{D}^\mathsf{opt},\mathbf{F}_\mathsf{cof}^\mathsf{opt}$ (Model II)}

\addplot [color=mycolor1, line width=1pt, mark=triangle*, mark options={solid, mycolor1}, mark size=1.6pt]
  table[row sep=crcr]{%
0	12.5287220982441\\
5	12.7482953220281\\
10	12.6181419001845\\
15	12.3863386388908\\
20	12.4050171368288\\
};
\addlegendentry{$\mathbf{F}_\mathsf{BB}^\mathsf{opt},\mathbf{F}_\mathsf{RF}^\mathsf{opt},\mathbf{F}_\mathsf{cof}^\mathsf{opt}$ (Model II)}

\addplot [color=mycolor2, dashed, line width=1pt, mark=diamond*, mark options={solid, rotate=90, mycolor2}, mark size=1.7pt]
  table[row sep=crcr]{%
0	10.2966571243027\\
5	10.2942303554799\\
10	9.81932625006407\\
15	9.53182417341764\\
20	9.18090946710688\\
};
\addlegendentry{$\mathbf{F}_\mathsf{D}^\mathsf{opt},\mathbf{F}_\mathsf{sel}^\mathsf{opt}$ (Model I, fictitious)}

\addplot [color=mycolor2, line width=1pt, mark=diamond*, mark options={solid, rotate=90, mycolor2}, mark size=1.7pt]
  table[row sep=crcr]{%
0	10.128899883146\\
5	10.0168474341234\\
10	9.52344158661131\\
15	9.27524347154408\\
20	8.98550816899711\\
};
\addlegendentry{$\mathbf{F}_\mathsf{BB}^\mathsf{opt},\mathbf{F}_\mathsf{RF}^\mathsf{opt},\mathbf{F}_\mathsf{sel}^\mathsf{opt}$ (Model I, fictitious)}

\addplot [color=mycolor3, dashed, line width=1pt, mark=*, mark options={solid, mycolor3}, mark size=1.3pt]
  table[row sep=crcr]{%
0	8.6846443146026\\
5	8.5078720892362\\
10	7.88253342455075\\
15	7.40018677186675\\
20	7.06347033782925\\
};
\addlegendentry{$\mathbf{F}_\mathsf{D}^\mathsf{opt},\mathbf{F}_\mathsf{sel}^\mathsf{opt}$ (Model I, hardware)}

\addplot [color=mycolor3, line width=1pt, mark=*, mark options={solid, mycolor3}, mark size=1.3pt]
  table[row sep=crcr]{%
0	8.59920502511286\\
5	8.37188696064733\\
10	7.57508900788611\\
15	7.12682513015362\\
20	6.89751266303764\\
};
\addlegendentry{$\mathbf{F}_\mathsf{BB}^\mathsf{opt},\mathbf{F}_\mathsf{RF}^\mathsf{opt},\mathbf{F}_\mathsf{sel}^\mathsf{opt}$ (Model I, hardware)}

\addplot [color=mycolor4, dashed, line width=1pt, mark=+, mark options={solid, mycolor4}, mark size=2.3pt]
  table[row sep=crcr]{%
0	7.92781513665755\\
5	7.85468444600944\\
10	7.40250763087375\\
15	7.02261615947918\\
20	6.67563043322889\\
};
\addlegendentry{$\mathbf{F}_\mathsf{D}^\mathsf{opt}$ (WMMSE)\! +\! fixed antennas}

\addplot [color=mycolor4, line width=1pt, mark=+, mark options={solid, mycolor4}, mark size=2.3pt]
  table[row sep=crcr]{%
0	7.81535828046652\\
5	7.73394163420928\\
10	7.08834023705056\\
15	6.74446475243393\\
20	6.53243235106539\\
};
\addlegendentry{$\mathbf{F}_\mathsf{BB}^\mathsf{opt},\mathbf{F}_\mathsf{RF}^\mathsf{opt}$ (WMMSE) \!\!\! + \!\!\! fixed antennas}

\end{axis}
\end{tikzpicture}%
    \vspace{-3em}
    \caption{{Evaluation of weighted sum-rate under imperfect CSI.}}
    \label{fig_sim10}
\end{figure}

A final aspect of interest is the robustness of the proposed algorithms to imperfect \ac{CSI}. In Fig.~\ref{fig_sim10}, we evaluate the weighted sum-rate achieved by precoders optimized using different algorithms under varying levels of \ac{CSI} errors. As modeled in Section~\ref{sec:algoM1}, we introduce additive estimation errors into the \acp{CSI} used for precoder optimization. For Model~I, we define the channel estimation error $\Delta \Hm_k^\mathsf{sel}$ as
\begin{equation}\notag
\Delta \Hm_k^\mathsf{sel} = \eta \frac{\|\bar{\Hm}_k^\mathsf{sel}\|_{\mathsf F}}{\sqrt{NSM_k}} \Nm_k,\ \forall k,
\end{equation}
where \(\bar{\Hm}_k^\mathsf{sel}\) is the ground-truth channel, \(\eta \ge 0\) controls the error level, and \(\Nm_k\) is a random matrix with i.i.d. entries following $\mathcal{CN}(0,1)$. The same perturbation model is applied to the effective channels in Model~II and the fixed-antenna case, with the normalization adjusted according to their matrix dimensions. In the simulations, the precoders are optimized using the perturbed channels, while the weighted sum-rate is evaluated on the ground-truth channels. Fig.~\ref{fig_sim10} shows that all considered algorithms exhibit a limited performance degradation within $\unit[2]{bps/Hz}$ when the \ac{CSI} error level is up to 20\%, indicating decent robustness to \ac{CSI} errors. 
}

\section{Conclusion}\label{sec:cons}
This paper investigates a tri-hybrid \ac{MIMO} system that integrates digital, analog, and antenna-domain precoding via pattern-reconfigurable antennas for multi-user communications. The primary contribution lies in developing two channel models that capture varying levels of radiation pattern flexibility, along with corresponding WMMSE-based precoding algorithms under practical power constraints and with quadratic complexity.  
The two developed models evaluate the system performance under two scenarios: (i) a practical setting constrained by existing pattern-reconfigurable antenna hardware, and (ii) an idealized case where arbitrary radiation patterns can be synthesized. The results allow us to answer the two questions posed in Section~I. For Q1, with the current hardware in~\cite{Wang2025Electromagnetically}, the achievable performance gain is modest, yielding approximately $\unit[0.9]{bps/Hz}$ in our simulation environment. However, with more advanced hardware supporting greater pattern flexibility, the performance gain can be substantial, exceeding $\unit[5.8]{bps/Hz}$. For Q2, the tri-hybrid architecture is shown to significantly reduce the number of RF chains and antennas without compromising system performance. Moreover, the higher the degree of pattern reconfigurability, the greater the potential for hardware cost savings.

Nonetheless, a key limitation of this work lies in the reliance of the proposed algorithm on an alternating optimization framework, which does not guarantee global optimality and remains computationally demanding. Future research directions include the development of more efficient optimization techniques for real-time applications, the design of more flexible antennas, {the joint design and control of reconfigurable antennas at both the BS and UE sides,} and the incorporation of more realistic effects such as mutual coupling and polarization.

\appendix
\setcounter{equation}{0}
\renewcommand\theequation{A.\arabic{equation}}

Based on Lemma~\ref{lemma:1}, the optimization of $\fv_{(n)}$ in~\eqref{eq:optfb} given a fixed $\bv_{(n)}$ can be rewritten as
\begin{subequations}\label{eq:optfbAppen}
	\begin{align}
		\min_{\fv_{(n)}}&\quad a\|\fv_{(n)}\|_2^2 + 2\mathrm{Re}\big(\fv_{(n)}^\HH\dv\big) \label{eq:objAppen}\\
		\mathrm{s.t.}&\quad \|\fv_{(n)}\|_2^2 \leq P_n,
	\end{align}
\end{subequations}
where $a = \bv_{(n)}^\TT\Bm_{nn}\bv_{(n)} > 0$ and $\dv = (\Qm_n - \Dm_n)\bv_{(n)}$. It can be inferred that the optimal solution $\fv_{(n)}^\mathsf{opt}$ must be parallel and opposite to $\dv$, i.e.,
\begin{equation}\label{eq:foptxd}
	\fv_{(n)}^\mathsf{opt} = -x\dv, \quad x \in \mathbb{R}_+.
\end{equation}
This can be proven by contradiction: suppose there exists an optimal solution $\fv_{(n)}^\mathsf{opt}$ to~\eqref{eq:optfbAppen} that does not conform to~\eqref{eq:foptxd}. Then, we have $(\fv_{(n)}^\mathsf{opt})^\HH\dv > -\|\fv_{(n)}^\mathsf{opt}\|_2\|\dv\|_2$. In this case, we can always construct a new vector $\tilde{\fv}_{(n)}^\mathsf{opt} = -\frac{\|\fv_{(n)}^\mathsf{opt}\|_2}{\|\dv\|_2}\dv$, which yields a lower objective value in~\eqref{eq:objAppen} than $\fv_{(n)}^\mathsf{opt}$. This contradiction implies that $\fv_{(n)}^\mathsf{opt}$ cannot be optimal, and hence~\eqref{eq:foptxd} must hold.

Substituting~\eqref{eq:foptxd} into~\eqref{eq:objAppen}, the optimization problem becomes
$
		\min_{x}\ ax^2 - 2x \label{eq:objAppen2},\ 
		\mathrm{s.t.}\ x \leq {\sqrt{P_n}}/{\|\dv\|_2},
$
whose optimal solution is given by
$
	x^\mathsf{opt} = \min\left(\frac{1}{a}, \frac{\sqrt{P_n}}{\|\dv\|_2}\right).
$
Combining this with~\eqref{eq:foptxd} gives the expression for $\fv_{(n)}^\mathsf{opt}$ in~\eqref{eq:fn_opt}.

\bibliography{references}
\bibliographystyle{IEEEtran}

\begin{IEEEbiography}[{\includegraphics[width=1in,height=1.25in,clip,keepaspectratio]{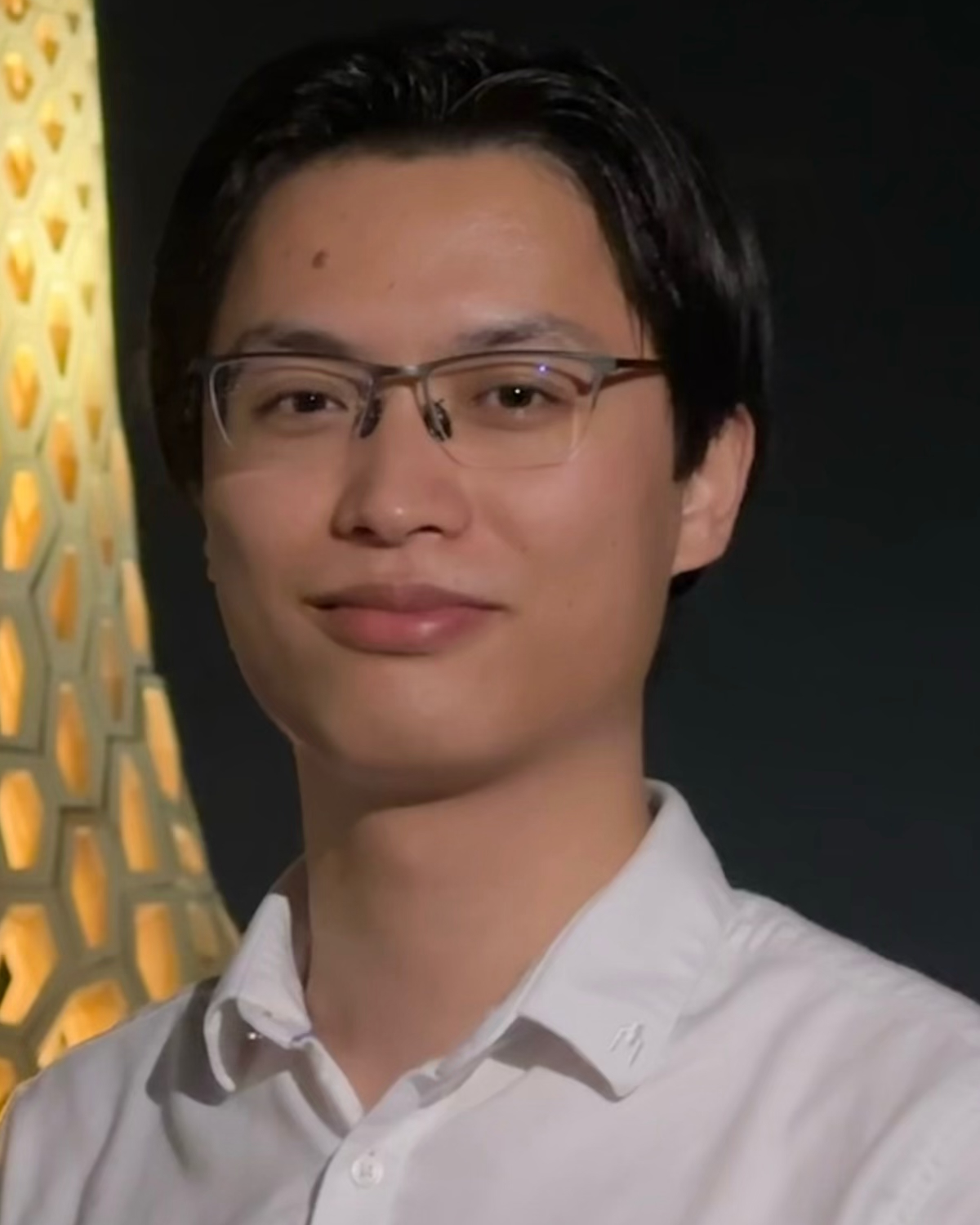}}]{Pinjun Zheng} (Member, IEEE) received the B.S. and M.S. degrees from Harbin Institute of Technology, Harbin, China, in 2019 and 2021, respectively, and the Ph.D. degree in electrical and computer engineering from the King Abdullah University of Science and Technology  (KAUST), Thuwal, Saudi Arabia, in 2024. He is currently a Postdoctoral Research Fellow at the School of Engineering, The University of British Columbia, Okanagan Campus, Kelowna, Canada. His research interests lie in signal processing, deep learning, and wireless communications. He received the Canada Postdoctoral Research Award (CPRA) from the Natural Sciences and Engineering Research Council of Canada (NSERC) in 2026. He is serving as an Associate Editor for IEEE Wireless Communications Letters. 
\end{IEEEbiography}

\begin{IEEEbiography}[{\includegraphics[width=1in,height=1.25in,clip,keepaspectratio]{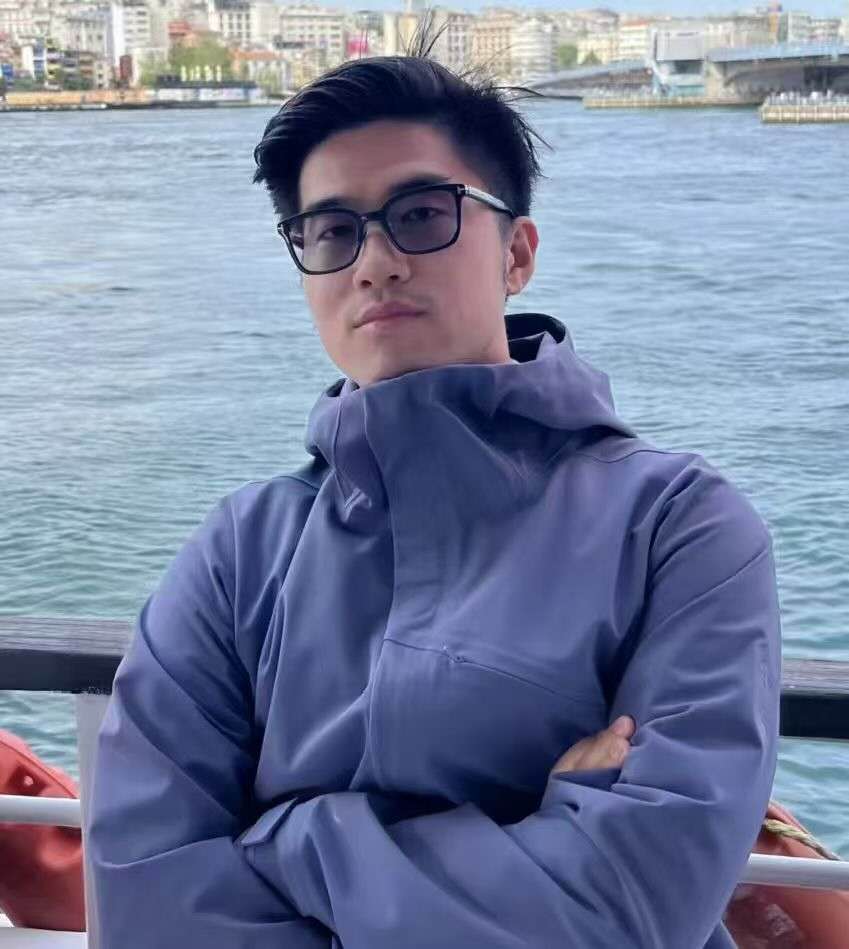}}]{Yuchen Zhang} (Member, IEEE) received the B.E. and Ph.D. degrees in communication engineering from the University of Electronic Science and Technology of China, Chengdu, China, in 2018 and 2024, respectively. From 2022 to 2023, he was a visiting Ph.D. student at the Weizmann Institute of Science, working under the guidance of Prof. Yonina C. Eldar. In 2024, he joined King Abdullah University of Science and Technology (KAUST), Thuwal, Saudi Arabia, as a Postdoctoral Researcher with Prof. Tareq Y. Al-Naffouri. In 2025, he was elevated to Global Fellow under the KAUST Global Fellowship Program, holding academic visiting positions at the University of Toronto, Canada; the University of Luxembourg, Luxembourg; and Imperial College London, U.K. His research interests include non-terrestrial networks, particularly low-Earth-orbit satellite systems, and efficient next-generation antenna/array architectures at the interface of wireless communications and electromagnetic/microwave theory for 6G and beyond. Dr. Zhang has served as a Technical Program Committee member for several IEEE flagship conferences, where he has also delivered tutorials and organized special sessions and workshops. He was recognized as an Exemplary Reviewer by IEEE Communications Letters. He currently serves as an Editor for IEEE Transactions on Communications and IEEE Communications Letters.
\end{IEEEbiography}

\begin{IEEEbiography}[{\includegraphics[width=1in,height=1.25in,clip,keepaspectratio]{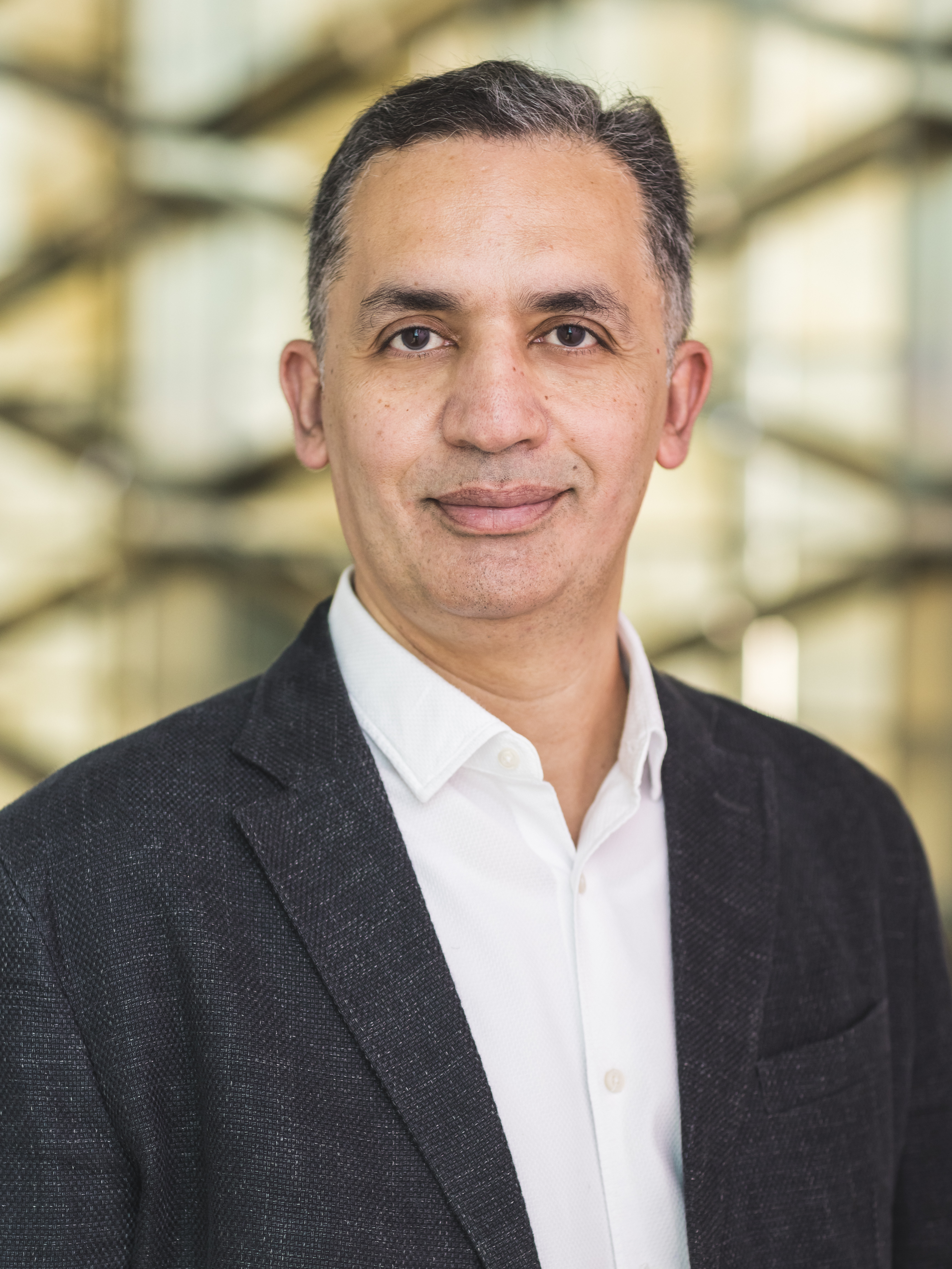}}]{Tareq Y. Al-Naffouri} (Fellow, IEEE) received the B.S. degrees in mathematics and electrical engineering (with first honors) from King Fahd University of Petroleum and Minerals, Saudi Arabia, the M.S. degree in electrical engineering from the Georgia Institute of Technology, and the Ph.D. degree in electrical engineering from Stanford University, Stanford in 2004. He was a visiting scholar at California Institute of Technology, Pasadena, CA in 2005 and summer 2006. He was a Fulbright scholar at the University of Southern California in 2008. He is currently a Professor at the Electrical Engineering Department, King Abdullah University of Science and Technology (KAUST). His research interests lie in the areas of sparse, adaptive, and statistical inference/learning and their applications to wireless communications, localization, smart cities, and smart health. He has over 370 publications in journal and conference proceedings and 24 issued/pending patents. He has won the IEEE Education Society Chapter Achievement Award (2008), Almarie Award for Innovative research in communication (2009), and Abdul Hameed Shoman Prize for innovative research in IoT (2022). He has been an IEEE Fellow since 2025. 
\end{IEEEbiography}

\begin{IEEEbiography}[{\includegraphics[width=1in,height=1.25in,clip,keepaspectratio]{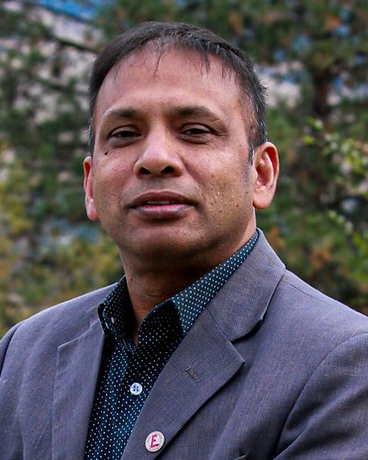}}]{Md. Jahangir Hossain} (S’04,M’08, SM’18) received the B.Sc. degree in electrical and electronics engineering
from Bangladesh University of Engineering and Technology (BUET), Dhaka, Bangladesh; the
M.A.Sc. degree from the University of Victoria, Victoria, BC, Canada, and the Ph.D. degree from the
University of British Columbia (UBC), Vancouver, BC, Canada.

He served as a Lecturer at BUET. He was a Research Fellow with McGill University, Montreal,
QC, Canada; the National Institute of Scientific Research, Quebec, QC, Canada; and the Institute for Telecommunications Research, University of South Australia, Mawson Lakes, Australia. His industrial experiences include a Senior Systems Engineer position with Redline Communications, Markham, ON, Canada, and a Research Intern position with Communication Technology Lab, Intel, Inc., Hillsboro, OR, USA. He is currently working as a Full Professor in the School of Engineering, UBC Okanagan campus, Kelowna, BC, Canada. His research interests include applications of machine learning for resource management in wireless networks, holographic multi-input and multi-output (HMIMO), and non-terrestrial communication systems. He served as  Editor for the IEEE Transactions on Communications and for the IEEE Transactions on Wireless Communications. He also served as an associate editor for the IEEE Communications Surveys and Tutorials. Dr. Hossain regularly serves as a member of the Technical Program Committee of the IEEE International Conference on Communications (ICC) and IEEE Global Telecommunications Conference (Globecom). 
\end{IEEEbiography}

\begin{IEEEbiography}[{\includegraphics[width=1in,height=1.25in,clip,keepaspectratio]{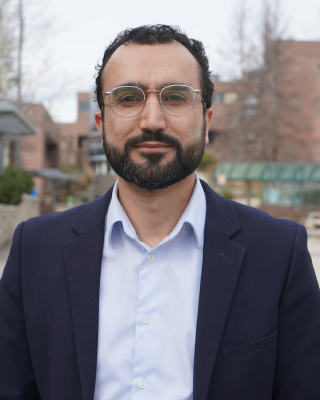}}]{Anas Chaaban} (S'09 - M'14 - SM'17) is an associate professor at the University of British Columbia. He received the Ma{\^i}trise {\`e}s Sciences degree in electronics from Lebanese University, Lebanon, in 2006, the M.Sc. degree in communications technology and the Dr. Ing. (Ph.D.) degree in electrical engineering and information technology from the University of Ulm and the Ruhr-University of Bochum, Germany, in 2009 and 2013, respectively. From 2008 to 2009, he was with the Daimler AG Research Group On Machine Vision, Ulm, Germany. He was a Research Assistant with the Emmy-Noether Research Group on Wireless Networks, University of Ulm, Germany, from 2009 to 2011, which relocated to the Ruhr-University of Bochum in 2011. He was a Post-Doctoral Researcher with the Ruhr-University of Bochum from 2013 to 2014, and with King Abdullah University of Science and Technology from 2015 to 2017. He joined the School of Engineering at the University of British Columbia in 2018. His research interests are in the areas of information theory and wireless communications.
\end{IEEEbiography}

\end{document}